\documentclass[10pt,conference,final,twocolumn]{IEEEtran}

\makeatletter\newif\if@restonecol
\makeatother

\usepackage{eucal}
\usepackage{graphicx}
\usepackage{times}
\usepackage{latexsym}
\usepackage{amsmath}
\usepackage{amssymb}
\usepackage{amsthm}
\usepackage{epsfig}
\usepackage{cite}
\usepackage{cases}
\usepackage{amsfonts}
\usepackage{indentfirst}
\usepackage{epsfig}
\usepackage{amsopn}
\usepackage{bm}
\usepackage[linesnumbered,ruled]{algorithm2e}
\usepackage{subfig}
\usepackage{eufrak}
\usepackage{stfloats}
\usepackage{booktabs}
\usepackage{slashbox}
\usepackage{arydshln}

\newtheorem{lemma}{Lemma}
\newtheorem{theorem}{Theorem}

\newcommand{\Tr}{\mathsf{Tr}}

\newcommand{\vect}{\mathsf{vec}}

\begin{document}

\title{Transceiver Design for Clustered Wireless Sensor Networks --- Towards SNR Maximization}
\author{Yang Liu \ \  Jing Li, \ \   Xuanxuan Lu \\
 Electrical and Computer Engineering Department,
 Lehigh University, Bethlehem, PA 18015\\
Email: \{yal210@lehigh.edu, jingli@ece.lehigh.edu, xul311@lehigh.edu\}
}



\maketitle

\footnotetext{Supported by National Science Foundation under Grants No. 0928092, 1133027 and 1343372.}

\begin{abstract}
This paper investigates the transceiver design problem in a noisy-sensing noisy-transmission multi-input multi-output (MIMO) wireless sensor network. Consider a cluster-based network, where multiple sensors scattering across several clusters will first send their noisy observations to their respective cluster-heads (CH), who will then forward the data to one common fusion center (FC). The cluster-heads and the fusion center collectively form a coherent-sum multiple access channel (MAC) that is affected by fading and additive noise. Our goal is to jointly design the linear transceivers at the CHs and the FC to maximize the signal-to-noise ratio (SNR) of the recovered signal. We develop three iterative block coordinated ascent (BCA) algorithms: 2-block BCA based on semidefinite relaxation (SDR) and rank reduction via randomization or solving linear equations, 2-block BCA based on iterative second-order cone programming (SOCP), and multi-block BCA that lends itself to efficient closed-form solutions in specific but important scenarios. We show that all of these methods optimize SNR very well but each has different efficiency characteristics that are tailored for different network setups. Convergence analysis is carried out and extensive numerical results are presented to confirm our findings.  



\end{abstract}


\section{Introduction}
\label{sec:introduction}

Wireless sensor networks (WSN) are known to have a broad range of applications such as environmental monitoring, battle-field surveillance and space exploration \cite{bib:WSNs_survey}. 
A typical wireless sensor network is composed of numerous sensors with each sensor having limited on-board processing and communication capability. The sensors are usually randomly casted over a large field in space to sense the same physical event. In practice the sensing observations tend to be disturbed due to the noise from hardware device or environment. The noisy observations are collected by the fusion centers (FC), where data fusion and further processing will be performed. The large amount of sensors are usually grouped into clusters, with each cluster formed by sensors located closely within a small neighborhood. Within each cluster, a cluser-head(CH) collects observations from other sensors free of error(this is reasonable since the cluster members lie in close vicinity) and then transmit the collected data to FC. Assume that HCs and FC are equipped with multiple antennas and linear transceivers, then a central question here is how to jointly design the transceivers to transmit the data reliably. This falls in the general problem of multi-input multi-output (MIMO) transceiver(beamforming) design problem, which has aroused a flurry of interest in recent years \cite{bib:FangLi_1, bib:FangLi_2, bib:GBG_1, bib:sensor_network_Cui, bib:sensor_network_Hamid, bib:Yang_ICC, bib:J1_Yang_to_be_submitted, bib:J2_Yang_to_be_submitted, bib:JunFang_MI}.

\begin{figure}[htb]
\centerline{
\includegraphics[width=0.5\textwidth]{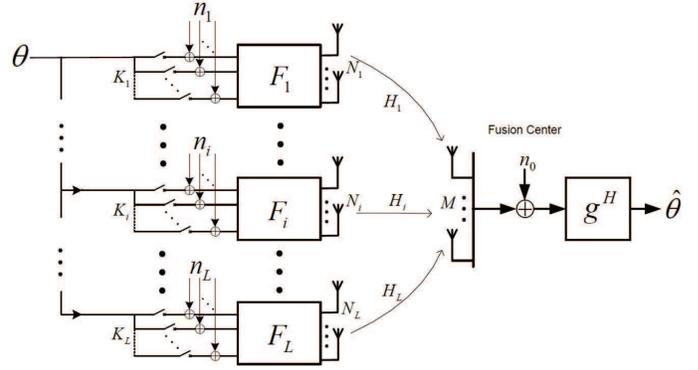}
}
\caption{Models for Cluster Based WSN}
\label{fig:sysmodel}
\end{figure}

We introduce the system model shown in Fig. \ref{fig:sysmodel} to capture the key characteristics of the afore-mentioned practical scenario. The event of interest $\theta$ is modeled as a complex scalar and sensed by a total of $\sum_{i=1}^{L}K_i$ sensors (cluster-members) coming from $L$ clusters. These sensors communicate their imperfect sensing results to their respective cluster-heads.  Within each cluster $i$, the observation distortion of the $K_i$ sensors are collectively modeled as an additive noise vector $\mathbf{n}_i$ (of size $K_i$), which may have any distribution (not necessarily Gaussian). The cluster-head $i$ performs linear precoding $\mathbf{F}_i$ to data before transmission with $N_i$ antennas. We assume that each cluster-head must conform to an individual transmission power constraint $P_i$, in accordance to its specific power supply or battery life. 
It should be noted that, compared to a single overall power constraint for all the cluster-heads, the set of separate power constraints makes the model not only more practical, but also considerably harder. 
The MIMO coherent-sum multiple access channel (MAC) is considered  
for communications between cluster-heads and fusion center, which achieves a high bandwidth efficiency. The fading channel between the $i$-th cluster-head and the fusion center is denoted as $\mathbf{H}_i$. 
At the fusion center, the received signal is corrupted by additive noise $\mathbf{n}_0$ (of size $M$). The fusion center performs linear postcoding $\mathbf{g}^H$ to deduce an estimate $\hat{\theta}$. Our goal is to jointly design the transceivers $\mathbf{F}_i$ and $\mathbf{g}$ to achieve a best estimate $\theta$. 

The line of beamforming design in WSN  research has actually received significant attention in recent years. A good variety of system setups have been investigated in the literature \cite{bib:FangLi_1, bib:FangLi_2, bib:GBG_1, bib:sensor_network_Cui, bib:sensor_network_Hamid, bib:Yang_ICC, bib:J1_Yang_to_be_submitted, bib:JunFang_MI, bib:J2_Yang_to_be_submitted}. Compared to these existing literature, our system model is generally more generic and complicated. 

For example, \cite{bib:FangLi_1} considers the case where each cluster has only one sensor (i.e. $K_i=1$, $\forall i$) and all cluster-heads and fusion center are equipped with one single antenna (i.e. $N_i=1$, $\forall i$ and $M=1$). Under the total power constraint over all the cluster-heads, the transceiver design problem boils down to a power allocation problem.
\cite{bib:FangLi_2} extends the scalar fading channel in \cite{bib:FangLi_1} to square nonsingular matrix channels, i.e. $N_i\!=\!M$, $\forall i=1,\cdots,L$. Still all clusters share one total power constraint. Note that the total power constraint over different clusters are usually nonrealistic since in practice different cluster-heads can be far away from each other without wired connections and are powered separately by build-in batteries. 
\cite{bib:GBG_1} extends the research by employing separate power constraints for each cluster-head, but still restricts the channel matrices $\mathbf{H}_i$ to be square and nonsingular. Block coordinate descent(BCD) algorithm is obtained in \cite{bib:GBG_1} to solve the problem.

Several studies are particularly noteworthy \cite{bib:sensor_network_Cui, bib:sensor_network_Hamid, bib:Yang_ICC, bib:J1_Yang_to_be_submitted, bib:JunFang_MI, bib:J2_Yang_to_be_submitted}. \cite{bib:sensor_network_Cui} is the first to propose a very general wireless sensor network model, where separate power constraints are employed and each fading channel matrix $\mathbf{H}_i$ can have arbitrary dimension($N_i>M$, $N_i<M$ or $N_i=M$). \cite{bib:sensor_network_Cui} studies solutions to several special but also rather meaningful cases including scalar channels, fading but noiseless channels and nonfading but noisy channels. 
\cite{bib:Yang_ICC} focuses on the scalar target source with nonscalar fading and noisy channels and obtains an approximate BCD algorithm with fully closed form solution for each step. The work \cite{bib:sensor_network_Hamid, bib:J1_Yang_to_be_submitted} develop variant BCD algorithms to tackle the most general case proposed in \cite{bib:sensor_network_Cui}, with convergence and closed form solutions being examined in \cite{bib:J1_Yang_to_be_submitted}. All the above works adopt mean square error(MSE) as performance metric. Recently joint transceiver design problems aiming mutual information(MI) maximization are studied in \cite{bib:JunFang_MI, bib:J2_Yang_to_be_submitted}, where orthogonal and coherent-sumq MAC are considered respectively.



The primary interest of this paper is to solve the joint transceiver design problem for the system depicted in Fig. \ref{fig:sysmodel}. The previous studies in\cite{bib:sensor_network_Cui,bib:sensor_network_Hamid,bib:Yang_ICC,bib:J1_Yang_to_be_submitted, 
bib:J2_Yang_to_be_submitted, bib:JunFang_MI} all used MSE or MI as the design criterion, and many of them assume a single power constraint (which renders an easier problem than separate power constraints). 
Here we take the signal-to-noise ratio (SNR) at the output of the FC postcoder $\mathbf{g}$ as a figure of merit. The SNR is an important performance indicator especially for discrete sources (where MSE becomes less meaningful). Maximizing SNR is equivalent to maximizing the symbol error rate (SER) in the discrete-source detection; and in the special case of Gaussian signaling, maximizing SNR is also equivalent to maximizing the channel capacity. The joint optimization problem is first formulated. Since it is non-convex, and since changing the design criterion (e.g. SNR instead of MSE) leads to a drastically different target function, the previous approaches and results can not be readily applied. Instead, we develop three feasible methods, all of which stem from the celebrated block coordinated ascent (BCA) principles, but each incorporates different techniques specifically tailored to the problem at hand, and renders computational strength in different scenarios.

The rest of this paper is organized as follows. The problem is formulated is Section \ref{sec:system model}. The optimal linear receiver is obtained in Section \ref{sec:opt_rec}.  Two different realizations of 2-block coordinated ascent (BCA) algorithm are discussed in Section \ref{sec:opt_all_sensors}, with their convergence being carefully examined. Section \ref{sec:multiple_BCA} proposes a multiple-block coordinate ascent approach, which promises tremendous reduction in complexity for special system settings. Extensive simulations are carried out to verify the effectiveness of the proposed algorithms in Section \ref{sec:numerical results} and Section \ref{sec:conclusion} concludes the whole paper.

\emph{Notations}: We use bold lowercase letters to denote complex vectors and bold capital letters to denote complex matrices. $\mathbf{0}$, $\mathbf{O}_{m\times n}$, and $\mathbf{I}_m$ are used to denote zero vectors, zero matrices of dimension $m\times n$, and identity matrices of order $m$ respectively. $\mathbf{A}^T$, $\mathbf{A}^{\ast}$ and $\mathbf{A}^H$ are used to denote transpose, conjugate and conjugate transpose(Hermitian transpose) respectively of an arbitrary complex matrix $\mathbf{A}$. $\Tr\{\cdot\}$ denotes the trace operation of a square matrix. $|\cdot|$ denotes the modulus of a complex scalar, and $\|\cdot\|_2$ denotes the $l_2$-norm of a complex vector. $\vect(\cdot)$ means vectorization operation of a matrix, which is performed by packing the columns of a matrix into a long one column. $\otimes$ denotes the Kronecker product. $\mathsf{Diag}\{\mathbf{A}_1,\cdots,\mathbf{A}_n\}$ denotes the block diagonal matrix with its $i$-th diagonal block being the square complex matrix $\mathbf{A}_i$, $i\in\{1,\cdots,n\}$. 
$\mathsf{Re}\{x\}$ and $\mathsf{Im}\{x\}$ denote the real and imaginary part of a complex value $x$, respectively.

\section{System Model}
\label{sec:system model}


We first discuss of the system model shown in Fig.\ref{fig:sysmodel} and present a mathematical formulation of the optimization problem at hand.
 

The system is composed of $\sum_{i=1}^{L}K_i$ sensors with $K_i$ sensors belonging to the $i$th cluster, $L$ sensor-heads each equipped with $N_i$ transmitting antennas, and a fusion center equipped with $M$ receiving antennas.  All sensors observe a common unknown source $\theta\in\mathbb{C}$, and send their observation to their respective cluster-heads. Without loss of generality, we assume $\theta$ has  zero mean and unit variance, i.e. $\mathsf{E}\{|\theta|^2\}=1$. 
 The sensing noise and transmission noise (to the $i$th cluster-head) can be collectively modeled as an additive noise $\mathbf{n}_i$, such that the $i$th cluster-head observes: 
\begin{align}
\label{eq:noisy_obs}
\mathbf{x}_i=\mathbf{1}_{K_i}\theta+\mathbf{n}_i, \ \ \ i=1,\cdots, L,
\end{align}
where $\mathbf{1}_{K_i}$ is a vector of dimension $K_i\times1$ with all the entries being $1$, and $\mathbf{n}_i$ denotes the additive noise vector with zero mean and covariance  matrix $\mathsf{E}\big\{\mathbf{n}_i\mathbf{n}_i^H\big\}=\mathbf{\Sigma}_i$, where $\mathbf{\Sigma}_i\succ0$. Following the convention of WSN, we assume that $\mathbf{n}_i$'s for the different clusters are mutually uncorrelated.  

The noisy observation $\mathbf{x}_i$ is beamformed by a linear precoder $\mathbf{F}_i\in\mathbb{C}^{N_i\times K_i}$ at each cluster-head before being sent to the fusion center. Let $\mathbf{H}_i\in\mathbb{C}^{M\times N_i}$ be the MIMO channel state information (CSI) from the $i$th cluster-head to the fusion center. Suppose all the cluster-heads form a coherent-sum MAC channel; then the fusion center   receives signal $\mathbf{r}$:
\begin{align}
\mathbf{r}&=\sum_{i=1}^{L}\Big(\mathbf{H}_i\mathbf{F}_i\mathbf{x}_i\Big)+\mathbf{n}_0\\
&=\sum_{i=1}^{L}\Big(\mathbf{H}_i\mathbf{F}_i\mathbf{1}_{K_i}\Big)\theta+\Big(\sum_{i=1}^{L}\mathbf{H}_i\mathbf{F}_i\mathbf{n}_i+\mathbf{n}_0\Big),
\end{align}
where $\mathbf{n}_0$ is the additive noise at the fusion center. Without loss of generality, we assume $\mathbf{n}_0$ has zero mean and white covariance matrix: $\mathsf{E}\{\mathbf{n}_0\mathbf{n}_0^H\}=\sigma_0^2\mathbf{I}_{M}$. Since the fusion center is usually far away from the sensing field, $\mathbf{n}_0$ is assumed uncorrelated with $\mathbf{n}_i$'s.


In practice, it is highly likely that each cluster-head is provisioned with a different power supply and hence must observe a  different transmission power constraint. The average transmission power for the $i$-th  cluster-head is $\mathsf{E}\Tr\{\mathbf{F}_i(\mathbf{1}_{K_i}\theta+\mathbf{n}_i)(\mathbf{1}_{K_i}\theta+\mathbf{n}_i)^H\mathbf{F}_i^H\}=\Tr\big\{\mathbf{F}_i\big(\mathbf{1}_{K_i}\mathbf{1}_{K_i}^H+\mathbf{\Sigma}_i\big)\mathbf{F}_i^H\big\}$, which must be no greater than a power limit $P_i$. 



The fusion center uses a linear postcoder $\mathbf{g}$ to perform data fusion and obtain an estimate $\hat{\theta}$: 
\begin{align}
\label{eq:received_theta}
\!\!\!\!\!\!\!\!\hat{\theta}\!\!=\!\!\mathbf{g}^H\mathbf{r}\!\!=\!\!\Big(\mathbf{g}^H\sum_{i=1}^{L}\mathbf{H}_i\mathbf{F}_i\mathbf{1}_{K_i}\Big)\theta\!+\!\mathbf{g}^H\Big(\sum_{i=1}^{L}\mathbf{H}_i\mathbf{F}_i\mathbf{n}_i\!\!+\!\!\mathbf{n}_0\Big).
\end{align}

The merit of the recovered signal $\hat{\theta}$ can be evaluated from several different perspectives. When the source $\theta$ takes value from a continuous space, the most popular metric is the mean square error, defined as $\mathsf{MSE}=\mathsf{E}\{|\theta-\hat{\theta}|^2\}$. MSE-targeted optimization has also been extensively studied in the beamforming literature (e.g. \cite{bib:FangLi_1, bib:FangLi_2, bib:GBG_1, bib:sensor_network_Cui, bib:sensor_network_Hamid, bib:Yang_ICC,  bib:J1_Yang_to_be_submitted}). 
Instead of taking MSE, here we target the average SNR, another important metric widely used in the design of communication systems. When the source $\theta$ is taken from from a finite discrete alphabet set (e.g. $M$-PAM or $M$-QAM), the detection accuracy is usually measured by the symbol error probability taking the form of $\mathsf{SER}\approx c_1\mathsf{Q}\big(\sqrt{\mathsf{c_2SNR}}\big)$, with $c_1$ and $c_2$ being some positive constants. When the source $\theta$ has Gaussian distribution, the system throughput is usually measured by the mutual information between $\theta$ and its estimate $\hat{\theta}$, $\mathsf{I}(\theta,\hat{\theta})=\frac{1}{2}\log_{2}(1+\mathsf{SNR})$. Hence, maximizing the SNR automatically minimizes the SER (for discrete source) or maximizes the mutual information (for Gaussian source).

From (\ref{eq:received_theta}), the  signal obtained by the fusion center (after linear postcoder $\mathbf{g}$) is composed of a signal component and a noise component, and the average SNR can be calculated and simplified to 
\begin{align}
\mathsf{SNR}\Big(\{\mathbf{F}\}_{i=1}^{L},\mathbf{g}\Big)&=\frac{\mathsf{E}\bigg\{\Big|\Big(\mathbf{g}^H\sum_{i=1}^{L}\mathbf{H}_i\mathbf{F}_i\mathbf{1}_{K_i}\Big)\theta\Big|^2\bigg\}}{\mathsf{E}\bigg\{\Big|\mathbf{g}^H\Big(\sum_{i=1}^{L}\mathbf{H}_i\mathbf{F}_i\mathbf{n}_i\!\!+\!\!\mathbf{n}_0\Big)\Big|^2\bigg\}} \\
&\!\!\!\!\!\!\!\!\!\!\!\!\!\!\!\!\!\!\!\!\!\!\!\!\!\!
=\frac{\mathbf{g}^H\big[\sum_{i=1}^L\mathbf{H}_i\mathbf{F}_i\mathbf{1}_{K_i}\big]\big[\sum_{j=1}^L\mathbf{H}_j\mathbf{F}_j\mathbf{1}_{K_j}\big]^H\mathbf{g}}{\sigma_0^2\|\mathbf{g}\|^2_2+\sum_{i=1}^{L}\mathbf{g}^H\mathbf{H}_i\mathbf{F}_i\mathbf{\Sigma}_i\mathbf{F}_i^H\mathbf{H}_i^H\mathbf{g}},\label{eq:SNR_func}
\end{align}
where the assumption of uncorrelated noise across the fusion center and different cluster-heads has been invoked. 

Hence, the joint transceiver design problem maximizing SNR for the clustered  wireless sensor network can be formulated as 
\begin{subequations}
\label{eq:original_SNR_prob}
\begin{align}
\!\!\!\!\!\!\!\!\!\!\!\!\!\!\!\!\!\!\!\!&(\mathsf{P}0):\underset{\{\mathbf{F}_i\}_{i=1}^L,\mathbf{g}\neq\mathbf{0}}{\mathsf{\max.}}\ \mathsf{SNR}\Big(\{\mathbf{F}\}_{i=1}^{L},\mathbf{g}\Big), \label{eq:original_SNR_prob_obj}\\
&\quad\ \mathsf{s.t.}\ \Tr\big\{\mathbf{F}_i\big(\mathbf{1}_{K_i}\mathbf{1}_{K_i}^H\!+\!\mathbf{\Sigma}_i\big)\mathbf{F}_i^H\big\}\!\leq\!P_i,i=1,\cdots,L.\label{eq:original_SNR_prob_constr}
\end{align}
\end{subequations}

The optimization problem ($\mathsf{P}0$) is nonconvex, as can be easily examined by checking the special case of scalar transceivers. 
In what follows, we will exploit the general principle of block coordinate ascent method to solve ($\mathsf{P}0$) in an iterative manner.

\section{Optimal Linear Receiver}
\label{sec:opt_rec}

In this section, we present the optimal linear receiver $\mathbf{g}$ which leads to the maximal SNR. The main result is as follows:
\begin{theorem}
\label{thm:opt_g}
For any predefined $\{\bm{F}_i\}_{i=1}^L$, SNR is maximized if and only if $\bm{g}^{\star}$ has the following form 
\begin{align}\label{eq:opt_g}
\mathbf{g}^{\star}\!\!=\!\alpha\bigg(\sigma_0^2\mathbf{I}_{M}\!\!+\!\!\sum_{i=1}^{L}\mathbf{H}_i\mathbf{F}_i\mathbf{\Sigma}_i\mathbf{F}_i^H\mathbf{H}_i^H\bigg)^{\!-\!1\!}\Big(\sum_{i=1}^{L}\mathbf{H}_i\mathbf{F}_i\mathbf{1}_{K_i}\Big),
\end{align}
where $\alpha$ is arbitrary nonzero complex scalar. The maximal SNR is given as 
\begin{align}\label{eq:opt_snr_g}
\!\!\!\!\!\!\!\!\!\mathsf{SNR}^{\star}\!\!=\!\!\Bigg\|\bigg(\!\sigma_0^2\mathbf{I}_{M}\!\!+\!\!\sum_{i=1}^{L}\mathbf{H}_i\mathbf{F}_i\mathbf{\Sigma}_i\mathbf{F}_i^H\mathbf{H}_i^H\!\bigg)^{\!-\!\frac{1}{2}}\!\!\!\Big(\!\sum_{i=1}^{L}\mathbf{H}_i\mathbf{F}_i\mathbf{1}_{K_i}\!\Big)\Bigg\|^2_2.
\end{align}
\end{theorem}
\begin{proof}
For simplicity we introduce the following notations
\begin{subequations}
\begin{align}
\mathbf{h}&\triangleq\sum_{i=1}^{L}\mathbf{H}_i\mathbf{F}_i\mathbf{1}_{K_i}; \\
\mathbf{M}&\triangleq\sigma_0^2\mathbf{I}_{M}\!\!+\!\!\sum_{i=1}^{L}\mathbf{H}_i\mathbf{F}_i\mathbf{\Sigma}_i\mathbf{F}_i^H\mathbf{H}_i^H. 
\end{align}
\end{subequations}
With all sensors' beamformers $\{\mathbf{F}_i\}_{i=1}^{L}$ given, the SNR maximization problem is the following optimization problem 
\begin{align}
\underset{\mathbf{g}\neq0}{\max.}\frac{\mathbf{g}^H\mathbf{h}\mathbf{h}^H\mathbf{g}}{\mathbf{g}^H\mathbf{M}\mathbf{g}}. 
\end{align}
Since $\mathbf{M}\succ0$, define $\tilde{\mathbf{g}}\triangleq\mathbf{M}^{\frac{1}{2}}\mathbf{g}$. The above problem becomes 
\begin{align}
\underset{\tilde{\mathbf{g}}\neq0}{\max.}\frac{\tilde{\mathbf{g}}^H\mathbf{M}^{-\frac{1}{2}}\mathbf{h}\mathbf{h}^H\mathbf{M}^{-\frac{1}{2}}\tilde{\mathbf{g}}}{\tilde{\mathbf{g}}^H\tilde{\mathbf{g}}}.
\end{align}
From variational perspective, the maximal value of the above fractional is obtained if and only if $\tilde{\mathbf{g}}$ is aligned with eigen-vector corresponding to the maximal eigenvalue of the matrix $\mathbf{M}^{-\frac{1}{2}}\mathbf{h}\mathbf{h}^H\mathbf{M}^{-\frac{1}{2}}$ \cite{bib:MatricAnalysis}. Notice that matrix $\mathbf{M}^{-\frac{1}{2}}\mathbf{h}\mathbf{h}^H\mathbf{M}^{-\frac{1}{2}}$ is rank-one and has only one positive eigenvalue whose eigen-vector is $\alpha\mathbf{M}^{-\frac{1}{2}}\mathbf{h}$, with $\alpha$ being any nonzero complex value. Thus the optimal solution of the above problem is $\tilde{\mathbf{g}}^{\star}=\alpha\mathbf{M}^{-\frac{1}{2}}\mathbf{h}$, from which (\ref{eq:opt_g}) and (\ref{eq:opt_snr_g}) can be readily obtained. 
\end{proof}
In practice the factor $\alpha$ can be chosen as 1 for convenience.

\section{Jointly Optimizing Beamformers at Sensors}
\label{sec:opt_all_sensors}

After obtaining the optimal linear receiver $\mathbf{g}$, we focus on optimizing precoders at the sensors' side in this section. 

First by utilizing the identities $\Tr\big\{\mathbf{A}\mathbf{B}\big\}=\Tr\big\{\mathbf{B}\mathbf{A}\big\}$ and $\Tr\big\{\mathbf{A}\mathbf{B}\mathbf{C}\mathbf{D}\big\}=\vect^H\big(\mathbf{D}\big)\big[\mathbf{C}^T\otimes\mathbf{A}\big]\vect\big(\mathbf{B}\big)$ \cite{bib:complex_matrix_derivative}, the numerator of $\mathsf{SNR}$ in (\ref{eq:SNR_func}) can be rewritten as follows
\begin{align}
\!\!\!\!\!\!&\qquad \mathbf{g}^H\bigg[\sum_{i=1}^L\mathbf{H}_i\mathbf{F}_i\mathbf{1}_{K_i}\bigg]\bigg[\sum_{j=1}^L\mathbf{H}_j\mathbf{F}_j\mathbf{1}_{K_j}\bigg]^H\mathbf{g}\nonumber\\
\!\!\!\!\!\!&=\sum_{i,j=1}^{L}\Tr\Big\{\big(\mathbf{H}_j^H\mathbf{g}\mathbf{g}^H\mathbf{H}_i\big)\mathbf{F}_i\big(\mathbf{1}_{K_i}\cdot\mathbf{1}_{K_j}^T\big)\mathbf{F}_j^H\Big\}\\
\!\!\!\!\!\!&\!=\!\sum_{i,j=1}^{L}\!\!\vect^H\big(\mathbf{F}_j\big)\bigg[\big(\mathbf{1}_{K_j}\!\!\cdot\!\!\mathbf{1}_{K_i}^T\big)\!\otimes\!\big(\mathbf{H}_j^H\mathbf{g}\mathbf{g}^H\mathbf{H}_i\big)\bigg]\vect\big(\mathbf{F}_i\big).
\end{align}
Similarly the denominator of $\mathsf{SNR}$ can be written as
\begin{align}
\!\!&\qquad\sigma_0^2\|\mathbf{g}\|^2_2+\sum_{i=1}^{L}\mathbf{g}^H\mathbf{H}_i\mathbf{F}_i\mathbf{\Sigma}_i\mathbf{F}_i^H\mathbf{H}_i^H\mathbf{g}\nonumber\\
\!\!&=\sum_{i=1}^{L}\Tr\Big\{\big(\mathbf{H}_i^H\mathbf{g}\mathbf{g}^H\mathbf{H}_i\big)\mathbf{F}_i\mathbf{\Sigma}_i\mathbf{F}_i^H\Big\}\!+\!\sigma_0^2\|\mathbf{g}\|^2_2\\
\!\!&=\vect^H\big(\mathbf{F}_i\big)\bigg[\mathbf{\Sigma}_i^*\otimes\big(\mathbf{H}_i^H\mathbf{g}\mathbf{g}^H\mathbf{H}_i\big)\bigg]\vect\big(\mathbf{F}_i\big)\!+\!\sigma_0^2\|\mathbf{g}\|^2_2,
\end{align}
and the $i$-th power constraint is expressed as 
\begin{align}
\!\!&\quad\Tr\Big\{\mathbf{F}_i\Big(\mathbf{1}_{K_i}\!\cdot\!\mathbf{1}_{K_i}^T\!\!+\!\!\mathbf{\Sigma}_i\Big)\mathbf{F}_i^H\Big\}\\
\!\!&=\vect^H\big(\mathbf{F}_i\big)\bigg[\Big(\big(\mathbf{1}_{K_i}\!\cdot\!\mathbf{1}_{K_i}^T\big)\!\!+\!\!\mathbf{\Sigma}_{i}^*\Big)\!\otimes\!\mathbf{I}_{N_i}\bigg]\vect\big(\mathbf{F}_i\big)\leq P_i.\nonumber
\end{align}
Here we introduce the following notations
\begin{subequations}
\label{eq:definition_fABCc}
\begin{align}
\!\!\!\!\!\!\mathbf{f}_i&\triangleq\vect\big(\mathbf{F}_i\big),\ \ i=1,\cdots,L;\\
\!\!\!\!\!\!\mathbf{A}_{ij}&\triangleq\Big[\big(\mathbf{1}_{K_i}\!\!\cdot\!\!\mathbf{1}_{K_j}^T\big)\otimes\big(\mathbf{H}_i^H\mathbf{g}\mathbf{g}^H\mathbf{H}_j\big)\Big],\ i,j\!=\!1,\cdots,L;\label{eq:definition_fABCc_A}\\
\!\!\!\!\!\!\mathbf{B}_{i}&\triangleq\Big[\mathbf{\Sigma}_i^*\otimes\big(\mathbf{H}_i^H\mathbf{g}\mathbf{g}^H\mathbf{H}_i\big)\Big],\ \ i=1,\cdots,L;\\
\!\!\!\!\!\!\mathbf{C}_i&\triangleq\Big[\Big(\big(\mathbf{1}_{K_i}\!\cdot\!\mathbf{1}_{K_i}^T\big)\!\!+\!\!\mathbf{\Sigma}_{i}^*\Big)\!\otimes\!\mathbf{I}_{N_i}\Big],\ \ i=1,\cdots,L;\\
\!\!\!\!\!\!c_0&\triangleq \sigma_0^2\|\mathbf{g}\|^2_2.
\end{align}
\end{subequations}
and define the matrix $\mathbf{A}\triangleq[\mathbf{A}_{ij}]_{i,j=1}^{L}$, (i.e. the $(i,j)$-th elementary block of $\mathbf{A}$ is $\mathbf{A}_{ij}$), $\mathbf{B}\triangleq\mathsf{Diag}\{\mathbf{B}_1,\cdots,\mathbf{B}_L\}$ and $\mathbf{D}_i\triangleq\mathsf{Diag}\{\mathbf{O}_{\sum_{j=1}^{i-1}K_jN_j}, \mathbf{C}_i,\mathbf{O}_{\sum_{j=i+1}^{L}K_jN_j}\}$ and pack all $\mathbf{f}_i$'s into one vector $\mathbf{f}\triangleq[\mathbf{f}_1^T,\cdots,\mathbf{f}_L^T]^T$. Then the problem of optimizing  beamformers $\{\mathbf{F}_i\}_{i=1}^{L}$ with $\mathbf{g}$ given is reformulated as follows
\begin{subequations}
\label{eq:opt_prob_P1}
\begin{align}
(\mathsf{P}1):\underset{\mathbf{f}}{\max.}&\ \frac{\mathbf{f}^H\mathbf{A}\mathbf{f}}{\mathbf{f}^H\mathbf{B}\mathbf{f}+c_0},\label{eq:opt_prob_P1_constr1}\\
\mathsf{s.t.}&\ \mathbf{f}^H\mathbf{D}_i\mathbf{f}\leq P_i,\ \ i\in\{1,\cdots,L\}.\label{eq:opt_prob_P1_constr2}
\end{align}
\end{subequations}

In the following we discuss methods solving the problem ($\mathsf{P}1$).

\subsection{Solving ($\mathit{P1}$) by Semidefinite Relaxation} 
\label{subsec:L_less_2}

In this subsection we solve problem ($\mathsf{P1}$) with help of recent progress in semidefinite relaxation techniques. 

First we rewrite the quadratic terms $\mathbf{f}^H\mathbf{A}\mathbf{f}$, $\mathbf{f}^H\mathbf{B}\mathbf{f}$ and $\mathbf{f}^H\mathbf{D}_i\mathbf{f}$ in ($\mathsf{P}1$) into inner-product forms $\Tr\big\{\mathbf{A}\mathbf{X}\big\}$, $\Tr\big\{\mathbf{B}\mathbf{X}\big\}$ and $\Tr\big\{\mathbf{D}_i\mathbf{X}\big\}$ respectively by introducing an intermediate variable $\mathbf{X}=\mathbf{f}\mathbf{f}^H$. Omitting this rank-one constraint, a relaxation of ($\mathsf{P}1$) is obtained as follows 
\begin{subequations}
\begin{align}
(\mathsf{P}2):\underset{\mathbf{X}}{\max.}&\ \frac{\Tr\big\{\mathbf{A}\mathbf{X}\big\}}{\Tr\big\{\mathbf{B}\mathbf{X}\big\}+c_0},\\
\mathsf{s.t.}&\ \Tr\big\{\mathbf{D}_i\mathbf{X}\big\}\leq P_i,\ \ i\in\{1,\cdots,L\},\\
&\ \mathbf{X}\succcurlyeq0.
\end{align}
\end{subequations}
 
The fractional semidefinite programming (SDP) objective in ($\mathsf{P}2$) is still nonconvex. To solve it, we utilize Charnes-Cooper's approach, which was originally proposed in \cite{bib:Charnes_Cooper} and subsequently adopted in many fractional SDP
optimization problems like \cite{bib:ADeMaio_TSP}\cite{bib:CJeong_TSP}. This turns ($\mathsf{P}2$) into the following SDP problem:
\begin{subequations}
\begin{align} 
(\mathsf{P}3): \underset{\mathbf{Y},\nu}{\max.}&\ \Tr\big\{\mathbf{A}\mathbf{Y}\big\}, \\
\mathsf{s.t.}&\ \Tr\big\{\mathbf{B}\mathbf{Y}\big\}+c_0\nu=1, \label{eq:P3_constr1}\\
&\ \Tr\big\{\mathbf{D}_i\mathbf{Y}\big\}\leq P_i\nu,\ i\in\{1,\cdots,L\}\label{eq:P3_constr2}\\
&\ \mathbf{Y}\succcurlyeq0, \nu\geq0.
\end{align}
\end{subequations}
The equivalence between ($\mathsf{P}2$) and ($\mathsf{P}3$) is established by the following lemma. 
\begin{lemma}
\label{lem:P2_P3_equivalent}
The problem ($\mathit{P2}$) and ($\mathit{P3}$) have equal optimal values. If $\bm{X}^{\star}$ solves ($\mathit{P2}$), then  $\left(\frac{\bm{X}^{\star}}{\Tr\{\bm{B}\bm{X}^{\star}+c_0\}},\frac{1}{\Tr\{\bm{B}\bm{X}^{\star}+c_0\}}\right)$
is an optimal solution to ($\mathit{P3}$). Conversely, if $(\bm{Y}^{\star},\nu^{\star})$ solves ($\mathit{P3}$), then $\nu^{\star}>0$ and $\bm{Y}^{\star}/\nu^{\star}$ solves ($\mathit{P2}$). 
\end{lemma}
\begin{proof}
See appendix \ref{subsec:appendix_lem_P2_P3_equivalent}.
\end{proof}

Since ($\mathsf{P}3$) is SDP problem, it can be solved by standard numerical solvers like CVX \cite{bib:CVX}. Remember that our goal is to solve problem ($\mathsf{P}1$). If the optimal solution $\mathbf{Y}^{\star}$ to ($\mathsf{P}3$) is rank-one, then the relaxation ($\mathsf{P}2$) is tight to ($\mathsf{P}1$) and the optimal solution of ($\mathsf{P}1$) can be obtained by eigenvalue decomposition of $\mathbf{Y}^{\star}/\nu{\star}$. When the optimal solution $\mathbf{Y}^{\star}$ has rank larger than one, constructing a solution to ($\mathsf{P}1$) from $(\mathbf{Y}^{\star},\nu^{\star})$ still needs to be addressed. 

To introduce our first major conclusion we need the following lemma.
\begin{lemma}
\label{lem:P3_solvability}
The problem ($\mathit{P3}$) and its dual are both solvable.
\end{lemma}
\begin{proof}
See appendix \ref{subsec:appendix_lem_P3_solvability}.
\end{proof}

For wireless sensor network with small number sensor clusters, we have the following conclusion.
\begin{theorem}
\label{thm:L_less_3}
If the wireless sensor network has no more than $3$ sensor clusters, i.e. $L\leq3$, then the relaxation ($\mathit{P2}$) is tight with respect to ($\mathit{P1}$). An optimal solution $(\bm{Y}^{\star},\nu^{\star})$ to ($\mathit{P3}$) with $\bm{Y}^{\star}$ being rank-one can be constructed and solution to ($\mathit{P1}$) can be obtained by eigenvalue-decomposing $\bm{Y}^{\star}/\nu^{\star}$.
\end{theorem}
\begin{proof}
The proof is inspired by theorem 3.2 of \cite{bib:YHuang_RankConstrained}. If $\mathbf{Y}^{\star}$ has rank one, nothing needs to be proved. Otherwise since the problem ($\mathsf{P}3$) and its dual ($\mathsf{D}3$) are both solvable by lemma \ref{lem:P3_solvability}, theorem 3.2 of \cite{bib:YHuang_RankConstrained} is valid to invoke. Define $r=\mathsf{rank}(\mathbf{Y}^{\star})$ and perform the following procedure:

\begin{itemize}
\item[-]\textbf{While} $\mathsf{rank}^2(\mathbf{Y}^{\star})+\mathsf{rank}(\nu^{\star})>L+1$ \textbf{Do}
\begin{quote}
\begin{itemize}
\item[Step-1:] Perform a full rank decomposition $\mathbf{Y}^{\star}=\mathbf{V}\mathbf{V}^H$, where $\mathbf{V}\in\mathbb{C}^{(\sum_{i=1}^{L}K_iN_i)\times r}$;
\item[Step-2:] Find a nonzero pair $(\mathbf{\Delta}, \delta)$, where $\mathbf{\Delta}$ is a $r\times r$ Hermitian matrix and $\delta$ is real scalar, such that the following linear equations are satisfied
\begin{align}
\!\!\!\!\!\!\!\!\!\!\Tr\big\{\mathbf{V}^H\mathbf{B}\mathbf{V}\mathbf{\Delta}\big\}\!+\!c_0\nu^{\star}\delta&\!=\!0;\label{eq:linear_eq1}\\
\!\!\!\!\!\!\!\!\!\!\Tr\big\{\mathbf{V}^H\mathbf{D}_i\mathbf{V}\mathbf{\Delta}\big\}\!-\!P_i\nu^{\star}\delta&\!=\!0, i\!=\!1,\cdots,L;\label{eq:linear_eq2}
\end{align}
\item[Step-3:] Evaluate $\kappa\!=\!\max\Big(\big|\lambda_{min}(\mathbf{\Delta})\big|, \big|\lambda_{max}(\mathbf{\Delta})\big|,|\delta|\Big)$;
\item[Step-4:] Update $\mathbf{Y}^{\star}=\mathbf{V}\big(\mathbf{I}_{\sum_{i=1}^{L}K_iN_i}-\kappa^{-1}\mathbf{\Delta}\big)\mathbf{V}^H$, $\nu^{\star}=\nu^{\star}(1-\kappa^{-1}\delta)$ and $r=\mathsf{rank}(\mathbf{Y}^{\star})$;
\end{itemize}
\end{quote}
\item[-]\textbf{End While}
\end{itemize}
In fact ($\mathsf{P3}$) has two semidefinite variables $\mathbf{Y}$ and $\nu$(note that $\nu$ is actually a nonnegative real scalar) and $L+1$ constraints. As long as the condition $\mathsf{rank}^2(\mathbf{Y}^{\star})+\mathsf{rank}(\nu)^{\star}>L+1$ holds, nonzero solutions to (\ref{eq:linear_eq1}) and (\ref{eq:linear_eq2}) exist. Thus after each repetition a new optimal solution is constructed with $\mathsf{rank}(\mathbf{Y}^{\star})$ being reduced by at least 1. Finally we obtain $\mathsf{rank}^2(\mathbf{Y}^{\star})+\mathsf{rank}(\nu)^{\star}\leq L+1$. Recall that $\nu^{\star}>0$ by lemma \ref{lem:P2_P3_equivalent}, so $\mathsf{rank}(\nu)^{\star}=1$ and we have $\mathsf{rank}^2(\mathbf{Y}^{\star})\leq L\leq 3$. So $\mathsf{rank}(\mathbf{Y}^{\star})=1$ and the theorem is proved.
%
\end{proof}

In the above, we have seen that ($\mathsf{P}1$) can be tackled by solving a SDP problem and then a finite number of linear equations when $L\leq3$. However the assumption that $L\leq3$ is still very stringent since in practice a sensor network can usually be composed of numerous clusters. A method to solve ($\mathsf{P1}$) suitable for arbitrary $L$ is still desirable. In the sequel, we proceed to discuss randomization method inspired by the recent literature \cite{bib:ZQLuo_SDR_Random}. Before going into details, first we modify the problem ($\mathsf{P}3$) a little bit. By changing the equality constraint (\ref{eq:P3_constr1}) into inequality, we have another SDP problem ($\mathsf{P}4$) as follows
\begin{subequations}
\begin{align} 
(\mathsf{P}4): \underset{\mathbf{Y},\nu}{\max.}&\ \Tr\big\{\mathbf{A}\mathbf{Y}\big\}, \\
\mathsf{s.t.}&\ \Tr\big\{\mathbf{B}\mathbf{Y}\big\}+c_0\nu\leq1, \label{eq:P4_constr1}\\
&\ \Tr\big\{\mathbf{D}_i\mathbf{Y}\big\}\leq P_i\nu,\ i\in\{1,\cdots,L\}\label{eq:P4_constr2}\\
&\ \mathbf{Y}\succcurlyeq0, \nu\geq0.
\end{align}
\end{subequations}
We assert that ($\mathsf{P3}$) and ($\mathsf{P4}$) are equivalent and for any solution $(\bm{Y}^{\star},\nu^{\star})$ to ($\mathsf{P}4$), $\nu^{\star}$ must be positive. In fact since ($\mathsf{P}4$) is a relaxation of ($\mathsf{P}3$), $\mathsf{opt}(\mathsf{P}4)\geq\mathsf{opt}(\mathsf{P}3)$. Conversely, if $(\mathbf{Y}^{\star},\nu^{\star})$ is an optimal solution to ($\mathsf{P}4$), then the constraint (\ref{eq:P4_constr1}) must indeed be active. Otherwise, $\mathbf{Y}^{\star}$ and $\nu^{\star}$ could be simultaneously inflated with a factor $\rho>1$ such that $(\rho\mathbf{Y}^{\star},\rho\nu^{\star})$ satisfies all constraints of ($\mathsf{P}4$) with (\ref{eq:P4_constr1}) being active and gives an strictly larger objective, which contradicts the optimality of $(\mathbf{Y}^{\star},\nu^{\star})$. So $(\mathbf{Y}^{\star},\nu^{\star})$ is feasible for ($\mathsf{P}3$) and thus $\mathsf{opt}(\mathsf{P}4)\leq\mathsf{opt}(\mathsf{P}3)$. Consequently ($\mathsf{P}3$) and ($\mathsf{P}4$) have equal optimal value. This means solution to either problem also solves the other one. Thus any solution $(\mathbf{Y}^{\star},\nu^{\star})$ to ($\mathsf{P}4$) is also a solution to ($\mathsf{P}3$) and by lemma 
\ref{lem:P2_P3_equivalent}, $\nu^{\star}>0$.


Assuming that we have obtained an optimal solution $(\mathbf{Y}^{\star},\nu^{\star})$ to ($\mathsf{P}4$), we now generate a sufficiently large number of independent complex random variables following the Gaussian distribution $\mathcal{CN}\big(\mathbf{0},\mathbf{Y}^{\star}\big)$. The intuition behind randomization comes from the observation of the following stochastic optimization problem 
\begin{subequations}
\begin{align} 
\!\!\!\!\!\!(\mathsf{P}5): \underset{\mathbf{f},\nu}{\max.}&\ \mathsf{E}_{\mathbf{f}\sim\mathcal{CN}(\mathbf{0},\mathbf{Y})}\big\{\mathbf{f}^H\mathbf{A}\mathbf{f}\big\}, \\
\!\!\!\!\!\!\mathsf{s.t.}&\ \mathsf{E}_{\mathbf{f}\sim\mathcal{CN}(\mathbf{0},\mathbf{Y})}\big\{\mathbf{f}^H\mathbf{B}\mathbf{f}\big\}+c_0\nu\leq1, \label{eq:P5_constr1}\\
\!\!\!\!\!\!&\ \mathsf{E}_{\mathbf{f}\sim\mathcal{CN}(\mathbf{0},\mathbf{Y})}\big\{\mathbf{f}^H\mathbf{D}_i\mathbf{f}\big\}\!\leq\!P_i\nu,i\!=\!1,\cdots,L,\label{eq:P5_constr2}\\
\!\!\!\!\!\!& \quad\nu\geq0.
\end{align}
\end{subequations}
By utilizing the relation $\mathsf{E}_{\mathbf{f}\sim\mathcal{CN}(\mathbf{0},\mathbf{Y})}\{\mathbf{f}\mathbf{f}^H\}=\mathbf{Y}$, the stochastic problem ($\mathsf{P}5$) actually becomes the SDP problem ($\mathsf{P}4$). The random variable $\mathbf{f}\sim\mathcal{CN}(\mathbf{0},\mathbf{Y}^{\star})$ solves the problem ($\mathsf{P}4$) in expectation. Thus if we have a sufficiently large number of samples, the ``best'' sample should solve the problem. 

The ``best'' sample can be found as follows. 
First, note that random samples are not always feasible for ($\mathsf{P}4$). This issue can be addressed by the following rescaling procedure. For each sample $\tilde{\mathbf{f}}$, we define the scaling factor $\beta(\tilde{\mathbf{f}})$ as
\begin{align}\label{eq:rescaling_1}
\beta\big(\tilde{\mathbf{f}}\big)=\underset{i=1,\cdots,L}{\min.}\left\{1,\sqrt{\frac{1-c_0\nu^{\star}}{\tilde{\mathbf{f}}^H\mathbf{B}\tilde{\mathbf{f}}}}, \sqrt{\frac{P_i\nu^{\star}}{\tilde{\mathbf{f}}^H\mathbf{D}_i\tilde{\mathbf{f}}}}\right\},
\end{align}
and rescale the sample $\tilde{\mathbf{f}}$ as  
\begin{align}\label{eq:rescaling_2}
\bar{\mathbf{f}}=\frac{\beta\big(\tilde{\mathbf{f}}\big)}{\sqrt{\nu^{\star}}}\tilde{\mathbf{f}}. 
\end{align}
It is easy to check that the obtained $\bar{\mathbf{f}}$ is guaranteed to be feasible for ($\mathsf{P}4$). Thus by performing the above rescaling procedure we can obtain a large number of feasible samples to approximate the optimal solution $\mathbf{Y}^{\star}$. Then we choose the one giving maximal objective value as solution to the problem ($\mathsf{P}1$). When the number of samples is sufficiently large, the obtained best objective value of ($\mathsf{P4}$) by rescaled random samples can be extremely close to true optimal value of ($\mathsf{P4}$) so the randomization solution can be regarded as tight to the original problem ($\mathsf{P}1$). 

In retrospect to the previous discussion, the motivation of transforming the problem ($\mathsf{P}3$) into its equivalent ($\mathsf{P}4$) now becomes clear. For implementation there is no chance that the randomly generated samples will satisfy the equality constraint (\ref{eq:P3_constr1}). At the same time the positivity of $\nu^{\star}$ guarantees that the rescaling in (\ref{eq:rescaling_2}) can be performed. 

Up to here, we have actually come out an alternative maximization method to solve the SNR optimization problem ($\mathsf{P}0$) in (\ref{eq:original_SNR_prob}). The algorithm starts from a random feasible point. In each iteration $\mathbf{g}$ is optimized in a closed form by theorem \ref{thm:opt_g} with $\{\mathbf{F}_i\}_{i=1}^L$ being fixed and $\{\mathbf{F}_i\}_{i=1}^L$ are optimized by solving ($\mathsf{P}3$) followed by randomization-rescaling or solving linear equations with $\mathbf{g}$ given. 

This algorithm is summarized in algorithm $\ref{alg:2BCA_SDR}$ as follows.

\begin{algorithm}
\caption{2-Block BCA to solve ($\mathsf{P}0$) using SDR and randomization)}
\label{alg:2BCA_SDR}
\textbf{Initialization}: Randomly generate nonzero feasible $\{\mathbf{F}_i^{(0)}\}_{i=1}^{L}$ such that $\mathbf{g}^{(0)}$ obtained by theorem \ref{thm:opt_g} is also nonzero; \ $j=0$\;
\Repeat{the increase of SNR becomes sufficiently small or a predefined number of iterations is reached}
{
 Solve ($\mathsf{P}3$) and obtain $(\mathbf{Y}^{\star},\nu^{\star})$\;
 \eIf{$L\leq3$} 
 {Reduce rank of $\mathbf{Y}^{\star}$ as in Theorem \ref{thm:L_less_3}; Obtain $\{\mathbf{F}_i^{(j+1)}\}_{i=1}^L$\;} 
 {Generate sufficiently large number of samples following $\mathcal{CN}\big(\mathbf{0},\mathbf{Y}^{\star}\big)$\;  
  Rescale each sample by (\ref{eq:rescaling_1}) and (\ref{eq:rescaling_2})\;
  Select among all rescaled samples the one giving maximal $\mathsf{SNR}$ as $\{\mathbf{F}_i^{(j+1)}\}_{i=1}^L$\;}
 Update $\mathbf{g}^{(j+1)}$ by theorem \ref{thm:opt_g};\ $j++$\;
}
\end{algorithm}

\subsection{Iteratively Solving ($\mathit{P1}$)}
\label{subsec:iterative_SOCP}

In the last subsection, we solve the problem ($\mathsf{P1}$) with help of semidefinite relaxation by first solving its SDP relaxation and than construct the rank-one solution through solving linear equations or randomization method.
In this subsection, we propose an alternative method which to solve ($\mathsf{P}1$) in an iterative manner. First we have the following conclusion
\begin{lemma}
\label{lem:A_rank1}
Matrix $\bm{A}$ in ($\mathit{P1}$) is rank-one. Specifically $\bm{A}=\bm{a}\bm{a}^H$ with the vector $\bm{a}$ being given as
\begin{align}\label{eq:def_a}
\bm{a}\triangleq\left[
\begin{array}{c}
\bm{1}_{K_1}\otimes\bm{H}_1^H\bm{g}\\
\vdots \\
\bm{1}_{K_L}\otimes\bm{H}_L^H\bm{g}\\
\end{array}
\right].
\end{align}
\end{lemma}
\begin{proof}
See appendix \ref{subsec:appendix_lem_A_rank1}.
\end{proof}

Now looking at the fractional SDP objective of ($\mathsf{P}1$) we have the following observation. For any given nonnegative real value $\gamma$, the SNR is no smaller than $\gamma$ is equivalent to the fact
\begin{align}\label{eq:SNR_larger_w}
\mathbf{f}^H\mathbf{A}\mathbf{f}\geq\gamma\mathbf{f}^H\mathbf{B}\mathbf{f}+\gamma c_0.
\end{align}
In other words, if $\mathsf{opt}(\mathsf{P}1)\geq \gamma$, then there exits some $\mathbf{f}$ such that the inequality (\ref{eq:SNR_larger_w}) and all power constraints $\mathbf{f}^H\mathbf{D}_i\mathbf{f}\leq P_i$ for $i=1,\cdots,L$ are simultaneously satisfied. If we define $u$ as follows 
\begin{align}
u\triangleq\underset{i=1,\cdots,L}{\max.}\left\{\frac{\mathbf{f}^H\mathbf{D}_i\mathbf{f}}{P_i}\right\},  
\end{align}
then the fact that all power constraints are satisfied is equivalent to $u\leq1$. Thus the statement $\mathsf{opt}(\mathsf{P}1)\geq\gamma$ holds if and only if the following optimization problem ($\mathsf{P}6_{\gamma}$) 
\begin{subequations}
\label{eq:frac_power_min_prob}
\begin{align}
(\mathsf{P}6_{\gamma}): \underset{\mathbf{f},u\geq0}{\min.}\ & u \label{eq:frac_power_min_prob_obj}\\
\mathsf{s.t.}\ & \mathbf{f}^H\mathbf{A}\mathbf{f}\geq\gamma\mathbf{f}^H\mathbf{B}\mathbf{f}+\gamma c_0, \label{eq:frac_power_min_prob_constr1}\\
& \frac{\mathbf{f}^H\mathbf{D}_i\mathbf{f}}{P_i}\leq u, i\in\{1,\cdots,L\}. \label{eq:frac_power_min_prob_constr2}
\end{align}
\end{subequations}
has optimal value smaller than $1$, i.e. $\mathsf{opt}(\mathsf{P}6_{\gamma})\leq1$.  

Next we show that all constraints of problem ($\mathsf{P}6_{\gamma}$) can be written in a second order cone form. The constraint (\ref{eq:frac_power_min_prob_constr1}), utilizing the result of lemma \ref{lem:A_rank1}, can be written as
\begin{align}\label{eq:SOCP_constr1}
\gamma\mathbf{f}^H\mathbf{B}\mathbf{f}+\gamma c_0\leq|\mathbf{a}^H\mathbf{f}|^2
\end{align}
Another key observation is that the optimal $\mathbf{f}^{\star}$ to ($\mathsf{P}6_{\gamma}$) is phase invariant---$(\mathbf{f}^{\star},u^{\star})$ is optimal solution to ($\mathsf{P}6_{\gamma}$) if and only if $(e^{\mathrm{j}\theta}\mathbf{f}^{\star},u^{\star})$ is optimal for any real value $\theta$. So without loss of optimality we assume that $\mathbf{a}^H\mathbf{f}=v$ with $v$ being a nonnegative real value. Thus the constraint (\ref{eq:SOCP_constr1}) readily becomes the second order cone
$\sqrt{\gamma}\big\|\big[\mathbf{f}^H\mathbf{B}^{\frac{1}{2}},\sqrt{c_0}\big]\big\|_2\leq v$. For the $i$-th power constraint in (\ref{eq:frac_power_min_prob_constr2}), it can also be written in a second order cone form $P_i^{\!-\!1/2}\|\mathbf{D}_i^{\frac{1}{2}}\mathbf{f}\|_2\leq u$. Thus the problem ($\mathsf{P}6_{\gamma}$) can be equivalently written in a standard second-order cone programming (SOCP) form:
\begin{subequations}
\label{eq:P7gamma}
\begin{align}
(\mathsf{P}7_{\gamma}):\ \underset{\mathbf{f},u,v}{\min.}& \ u, \\
\mathsf{s.t.}& \ 
\left\|\left[
\begin{array}{cc}
\sqrt{\gamma}\mathbf{B}^{\frac{1}{2}} & \mathbf{0}\\
\mathbf{0}^T & \sqrt{\gamma c_0}
\end{array}
\right]
\left[
\begin{array}{c}
\mathbf{f} \\
1
\end{array}
\right]\right\|_2\leq v, \\
& \mathsf{Re}\{\mathbf{a}^H\mathbf{f}\}=v, \\
& \mathsf{Im}\{\mathbf{a}^H\mathbf{f}\}=0, \\
& \Big\|\sqrt{P_i^{-1}}\mathbf{D}_i^{\frac{1}{2}}\mathbf{f}\Big\|_2\leq u,\ i=1,\cdots,L.
\end{align}
\end{subequations}
Thus if we know that the $\mathsf{opt}(\mathsf{P}1)$ lives in some interval, then $\mathsf{opt}(\mathsf{P}1)$ can be determined by a bisection search---we set $\gamma$ as middle point of the current search interval, if ($\mathsf{opt}(\mathsf{P}7_{\gamma})\leq1$), then $\mathsf{opt}(\mathsf{P}1)$ can achieve higher value and $\gamma$ is a lower bound of $\mathsf{opt}(\mathsf{P}1)$. Otherwise $\gamma$ upper-bounds $\mathsf{opt}(\mathsf{P}1)$.

Now the remaining problem is to determine an interval containing $\mathsf{opt}(\mathsf{P}1)$, from which the bisection search can start with. Since ($\mathsf{P}1$) is maximization problem, any feasible solution gives a lower bound of $\mathsf{opt}(\mathsf{P}1)$. The following lemma provides an upper bound of $\mathsf{opt}(\mathsf{P}1)$.
\begin{lemma}
\label{lem:opt_P1_upperbound}
Optimal value of ($\mathit{P1}$) has an upper bound as follows
\begin{align}
\label{eq:P1_upperbound}
\mathsf{opt}(\mathsf{P}1)\leq c_0^{-1}\left(\sum_{i=1}^{L}K_i\sqrt{\frac{P_i}{\lambda_{min}(\mathbf{C}_i)}}\big\|\mathbf{H}_i^H\mathbf{g}\big\|_2\right)^2. 
\end{align}
\end{lemma}
\begin{proof}
See appendix \ref{subsec:appendix_lem_opt_P1_upperbound}.
\end{proof}

Thus we have obtained an alternative method to solve the original problem ($\mathsf{P}0$), which also falls in the 2-block BCA framework. The steps are summarized in Algorithm \ref{alg:2BCA_SOCP}.

\begin{algorithm}
\caption{2-Block BCA to solve ($\mathsf{P}0$) based on SOCP}
\label{alg:2BCA_SOCP}
\textbf{Initialization}: Randomly generate nonzero feasible $\{\mathbf{F}_i^{(0)}\}_{i=1}^{L}$ such that $\mathbf{g}^{(0)}$ obtained by Theorem \ref{thm:opt_g} is also nonzero; \ $j=0$\;
\Repeat{the increase of SNR is sufficiently small or a predefined number of iterations is reached}
{
 Obtain $bd_l=\mathsf{SNR}\big(\{\mathbf{F}_i^{(j)}\}_{i=1}^L,\mathbf{g}^{(j)}\big)$ and $bd_u$ by (\ref{eq:P1_upperbound})\;
 \Repeat{$(bd_u-bd_l)$ is small enough}
 {
  Set $\gamma=(bd_u+bd_l)/2$; solve ($\mathsf{P}7_{\gamma}$)\; 
  \eIf{$\mathsf{opt}(\mathsf{P}7_{\gamma})\leq1$}
  {$bd_l=\gamma$\;}
  {$bd_u=\gamma$\;}
 }
 $\gamma=bd_l$\;
 Solve ($\mathsf{P}7_{\gamma}$) to update $\{\mathbf{F}_i^{(j+1)}\}_{i=1}^{L}$\; 
 Update $\mathbf{g}^{(j+1)}$ by theorem \ref{thm:opt_g};\ $j++$\;
}
\end{algorithm}

\subsection{Convergence and Complexity Analysis}
\label{subsec:convergence}

The two 2-block BCA algorithms  developed in the previous subsections have the following convergence property:
\begin{theorem}
\label{thm:2BCA_convergence}
The sequence of SNR obtained by algorithm \ref{alg:2BCA_SDR} or \ref{alg:2BCA_SOCP} converges. Moreover the solution sequence generated by algorithm \ref{alg:2BCA_SDR} or \ref{alg:2BCA_SOCP} has limit points and each limit point is a stationary point of problem ($\mathit{P0}$).
\end{theorem}
\begin{proof}
See appendix \ref{subsec:appendix_thm_2BCA_convergence}. 
\end{proof}


The complexity of the proposed algorithms is complicated since the whole network has too many factors ($K_i$'s and $N_i$'s) that impact the problem size. To simplify the analysis, we consider homogeneous sensor networks, where each cluster has the same number of sensors and each cluster-head has the same the number of antennas, i.e. $K_i=K$ and $N_i=N$ for all $i=1,\cdots, L$.

Using the primal-dual interior point method \cite{bib:Lectures_BenTal}, we can show that the complexity to solve ($\mathsf{P}4$) is $\mathcal{O}\big(L^{4.5}K^{3.5}N^{3.5}\big)$. The complexity to update $\mathbf{g}$ by (\ref{eq:opt_g}) comes primarily from the matrix inversion operation, which has a complexity $\mathcal{O}\big(M^3\big)$. 
In general, a few thousand of random samples are sufficient to guarantee a satisfying tightness of the obtained rank-reduced solutions (usually within $10^{-4}$ from the true optimal value), and the required number of samples does not increase with the network size. Thus, the complexity for each outer-layer iteration of the SDP-based 2BCA algorithm is $\mathcal{O}\big(L^{4.5}K^{3.5}N^{3.5}\!+\!M^3\big)$. 

From \cite{bib:IPM_Terlaky}, the complexity for solving the SOCP problem ($\mathsf{P}7_{\gamma}$) is $\mathcal{O}\big(L^{3.5}K^3N^3\big)$. Recall that each round of bisection search solves ($\mathsf{P}7_{\gamma}$) once, so ($\mathsf{P}7_{\gamma}$) is solved multiple times within one outer-layer iteration. Taking different channel conditions and levels of predefined precision into account, numerical results show that the number of times solving ($\mathsf{P}7_{\gamma}$) varies between the narrow range $[25,35]$ and thus can be considered as a constant. Thus the complexity of outer-layer SOCP-based 2BCA algorithm is $\mathcal{O}\big(L^{3.5}K^3N^3\!+\!M^3\big)$.

\section{Multiple Block Framework to Maximize SNR}
\label{sec:multiple_BCA}

In the previous sections, the proposed algorithms are both 2-block coordinate ascent methods where all the beamformers' sensors are jointly updated. One problem for these algorithms is that the complexity of solving the associated SDP or SOCP problem grows intensively with the increase of the size of the wireless sensor network.
Instead of jointly optimizing all beamformers, we can alternatively focus on just one sensor's beamformer each time. This actually results in a multiple-block BCA approach, which, despite the many subproblems involves, often involves a lower complexity (see complexity analysis and numerical results). Specifically, for the case of $K_i\!=\!1$ (1 sensor in the $i$-th cluster), the solution to the $i$th block is quickly obtained in a closed form, requiring no numerical solvers.

Now we consider the problem of optimizing the $i$-th beamformer $\mathbf{F}_i$ with $\mathbf{g}$ and $\{\mathbf{F}_j\}_{j\neq i}$ being fixed. By introducing the following notations 
\begin{subequations}
\label{eq:def_qcd}
\begin{align}
\mathbf{q}_i&\triangleq\!\!\sum_{j\neq i}\mathbf{A}_{ij}\mathbf{f}_j;\ c_i\!\triangleq \!\sum_{j,k\neq i}\mathbf{f}_j^H\mathbf{A}_{jk}\mathbf{f}_k; \\
d_i&\triangleq\sigma_0^2\|\mathbf{g}\|_2^2\!+\!\sum_{j\neq i}\mathbf{f}_j\mathbf{B}_j\mathbf{f}_j,
\end{align}
\end{subequations}
this problem is formulated as follows
\begin{subequations}
\begin{align}
(\mathsf{P}1^i): \underset{\mathbf{f}_i}{\max.}&\ \frac{\mathbf{f}_i^H\mathbf{A}_{ii}\mathbf{f}_i\!+\!2\mathsf{Re}\{\mathbf{q}_i^H\mathbf{f}_i\}+c_i}{\mathbf{f}_i^H\mathbf{B}_i\mathbf{f}_i\!+\!d_i}, \\
\mathsf{s.t.}&\  \mathbf{f}_i^H\mathbf{C}_i\mathbf{f}_i\leq P_i.
\end{align}
\end{subequations}

\subsection{One-Shot SDR-Rank-Reduction Method}
\label{subsec:one-shot}

First we introduce a one-shot method to solve ($\mathsf{P}1^i$), which performs semidefinite programming and rank-one matrix decomposition in tandem. This method is discussed in recent work \cite{bib:ADeMaio_TSP} and \cite{bib:CJeong_TSP}.  

By use of Charnes-Cooper's transformation and rank-one relaxation we turn ($\mathsf{P}1^i$) into the following relaxed version 
\begin{subequations}
\begin{align}
\!\!\!\!\!\!\!\!\!\!\!\!\!\!\!\!(\mathsf{P}7^i):\ \underset{\mathbf{Z},\eta}{\max.}&\ \Tr\big\{\mathbf{Q}_1\mathbf{Z}\big\}, \\
\!\!\!\!\!\!\!\!\!\!\!\!\!\!\!\!\mathsf{s.t.}&\ \Tr\big\{\mathbf{Q}_2\mathbf{Z}\big\}=1,\\
\!\!\!\!\!\!\!\!\!\!\!\!\!\!\!\!&\ \Tr\big\{\mathbf{Q}_3\mathbf{Z}\big\}\leq P_i\eta, \\
\!\!\!\!\!\!\!\!\!\!\!\!\!\!\!\!&\ \Tr\big\{\mathbf{Q}_4\mathbf{Z}\}=\eta, \\
\!\!\!\!\!\!\!\!\!\!\!\!\!\!\!\!&\ \mathbf{Z}\succcurlyeq0, \eta\geq0.
\end{align}
\end{subequations}
with parameter matrices being defined as
\begin{subequations}
\label{eq:def_Q1234}
\begin{align}
\mathbf{Q}_1&\triangleq\left[
\begin{array}{cc}
\mathbf{A}_{ii} & \mathbf{q}_i \\
\mathbf{q}_i^H & c_i \\
\end{array}
\right], 
\mathbf{Q}_2\triangleq\left[
\begin{array}{cc}
\mathbf{B}_{i} & \mathbf{0} \\
\mathbf{0}^T & d_i \\
\end{array}
\right], \\
\mathbf{Q}_3&\triangleq\left[
\begin{array}{cc}
\mathbf{C}_{i} & \mathbf{0} \\
\mathbf{0}^T & 0 \\
\end{array}
\right], 
\mathbf{Q}_4\triangleq\left[
\begin{array}{cc}
\mathbf{O} & \mathbf{0} \\
\mathbf{0}^T & 1 \\
\end{array}
\right].
\end{align}
\end{subequations}
Solving the SDP ($\mathsf{P}7^i$) we obtain an solution $(\mathbf{Z}^{\star},\eta^{\star})$. If the $\mathbf{Z}^{\star}$ is rank one, i.e. $\frac{\mathbf{Z}^{\star}}{\eta^{\star}}=\mathbf{z}^{\star}\mathbf{z}^{\star H}$ with $\mathbf{z}^{\star}\triangleq[\mathbf{z}_1^{T},z_2]^T$, then $\mathbf{z}_1^{\star}/z_2$ is an solution to ($\mathsf{P}1^i$) and the relaxation ($\mathsf{P}7^i$) is actually tight with respect to ($\mathsf{P}1^i$). Actually the rank-one solution $\mathbf{Z}^{\star}$ always exits due to the recent matrix decomposition result in \cite{bib:rank_one_decomp}. In fact if $\mathbf{Z}^{\star}$ has rank larger than one, by help of theorem 2.2 in \cite{{bib:rank_one_decomp}}, we can obtain a vector $\mathbf{z}$ such that the equations $\Tr\{(\mathbf{Q}_1\!-\!\mathsf{opt}(\mathsf{P}7^i)\mathbf{Q}_2)\mathbf{z}\mathbf{z}^H\}=0$, $\Tr\{\mathbf{Q}_j\mathbf{z}\mathbf{z}^H\}=\Tr\{\mathbf{Q}_j\mathbf{Z}^{\star}\}$ for $j=3,4$. This means $(\mathbf{z}\mathbf{z}^H,\eta^{\star})$ is rank-one optimal solution to ($\mathsf{P}7^i$) and thus ($\mathsf{P}1^i$) can be solved. 


\subsection{Iterative Method}
\label{subsec:iterative_method_Prob_i}

Besides the above one-shot method, here we propose an alternative iterative method to solve ($\mathsf{P}1^i$). As we will shortly see this iterative method can give birth to extremely efficient solution to ($\mathsf{P}1^i$) in specific circumstance. 

For any given positive real value $\alpha$, the fact that the SNR objective in ($\mathsf{P}1^i$) is no smaller than $\alpha$ equivalently reads 
\begin{align}
\mathbf{f}_i^H\big[\alpha\mathbf{B}_i\!-\!\mathbf{A}_{ii}\big]\mathbf{f}_i\!-\!2\mathsf{Re}\big\{\mathbf{q}_i^H\mathbf{f}_i\big\}\!+\!(\alpha d_i\!-\!c_i)\leq0. 
\end{align}
This immediately implies that if the following problem with
\begin{subequations}
\label{eq:opt_prob_P_i}
\begin{align}
\!\!(\mathsf{P}8_{\alpha}^i):\underset{\mathbf{f}_i}{\min.}&\ \mathbf{f}_i^H\big[\alpha\mathbf{B}_i\!\!-\!\!\mathbf{A}_{ii}\big]\mathbf{f}_i\!\!-\!2\mathsf{Re}\big\{\mathbf{q}_i^H\mathbf{f}_i\big\}\!\!+\!\!(\alpha d_i\!\!-\!\!c_i), \\
\!\!\mathsf{s.t.}&\ \mathbf{f}_i^H\mathbf{C}_i\mathbf{f}_i\leq P_i. 
\end{align}
\end{subequations}
with $\alpha$ given has a nonnegative optimal value then $\mathsf{opt}(\mathsf{P}1^i)\geq\alpha$. Otherwise $\alpha$ can serve as an upper bound of $\mathsf{opt}(\mathsf{P}1^i)$. Thus we can perform a bisection search to solve ($\mathsf{P}1^i$). Now the problem reduces to how to solve the problem ($\mathsf{P}8_{\alpha}^i$)? Note that the quadratic matrix $\big[\alpha\mathbf{B}_i\!-\!\mathbf{A}_{ii}\big]$ can be negative semidefinite or indefinite and thus ($\mathsf{P}8_{\alpha}^i$) is possibly nonconvex. The following theorem convinces us that ($\mathsf{P}8_{\alpha}^{i}$) can always be solved regardless of the convexity of its objective.
\begin{theorem}
\label{thm:solvability_S_pros}
If the $i$-th sensor cluster has more than one sensor or the head is equipped with multiple antenna, i.e. $K_i\geq 2$ or $N_i\geq 2$, the problem ($\mathit{P8_{\alpha}^i}$) can be solved.
\end{theorem}
\begin{proof}
See appendix \ref{subsec:appendix_thm_solvability_S_pros}.
\end{proof}

Although theorem \ref{thm:solvability_S_pros} shows that the problem ($\mathsf{P}8_{\alpha}^i$) can be solved by SDR and thus the iterative method to solve ($\mathsf{P}1^i$) works, it is generally less efficient than the one-shot method discussed above. Since the former performs semidefinite programming and rank-one reductions multiple times while the latter for just once. However in the circumstance where $K_i=1$, the following theorem indicates that ($\mathsf{P}8_{\alpha}^i$) has fully closed form solution and consequently the iterative method can become extremely efficient.

\begin{theorem}
\label{thm:closedform_P_i_D1}
When $K_i=1$, denote $\mathbf{\Sigma}_{\mathbf{s}}$ and $\{\mathbf{\Sigma}_i\}_{i=1}^L$ as scalars $\sigma_s^2$ and $\{\sigma_i^2\}_{i=1}^L$ respectively.
The solution $\bm{f}_i^{\star}$ to ($\mathit{P8_{\alpha}^i}$) is
\begin{align}
\!\!\!\!\!\!\bm{f}_i^{\star}\!\!=\!\!\left\{
\begin{array}{cl}
\frac{\beta_i^*\mathbf{H}_i^H\mathbf{g}}{\|\mathbf{H}_i^H\mathbf{g}\|_2^2(\alpha\sigma_i^2\!-\!1)}, & \text{if}\ \alpha\sigma_i^2\!>\!1,\frac{|\beta_i^{\ast}|}{\|\mathbf{H}_i^H\mathbf{g}\|_2(\alpha\sigma_i^2\!-\!1)}\!\leq\!\sqrt{\bar{P}_i};\\
\frac{\sqrt{\bar{P}_i}\beta_i^*\mathbf{H}_i^H\mathbf{g}}{|\beta_i|\|\mathbf{H}_i^H\mathbf{g}\|_2}, & \text{otherwise},
\end{array}
\right.\nonumber
\end{align}
with $\beta_i$ and $\bar{P}_i$ being defined as follows
\begin{align}
\label{eq:def_P_bar}
\beta_i&\triangleq\sum_{j\neq i}\mathbf{f}_j^H\Big(\mathbf{1}_{K_j}\otimes\mathbf{H}_j^H\mathbf{g}\Big),\ \bar{P}_i\triangleq\frac{P_i}{\sigma_i^2+\sigma_s^2},
\end{align}
\end{theorem}
\begin{proof}
See appendix \ref{subsec:appendix_thm_closedform_P_i_D1}.
\end{proof}

Note that the closed-form solution in Theorem \ref{thm:closedform_P_i_D1} neither requires matrix decomposition or solving linear equations(matrix inversion) nor depends on numerical solver. Thus iteratively solving ($\mathsf{P}1^i$) is easy for implementation and has very low computation cost.  Comparatively, the one-shot method for solving ($\mathsf{P}1^i$) depends on numerical solvers (like CVX) which are iterative-based (interior point method) solvers with each iteration performing matrix decomposition and solving linear equations. 

To start the bisection search, the latest SNR can serve as a lower bound for $\mathsf{opt}(\mathsf{P}1^i)$. From (\ref{eq:bd_norm_f_i}), we can derive an upper bound for $\mathsf{opt}(\mathsf{P}1^i)$:  
\begin{align}
\label{eq:bd_P_1_i}
\!\!\!\!\!\!d_i^{-1}\Big(\|\mathbf{H}_i^H\mathbf{g}\|_2^2P_i\lambda_{min}^{-1}\big(\mathbf{C}_i\big)\!\!+\!\!2\|\mathbf{q}_i\|_2P_i^{\frac{1}{2}}\lambda_{min}^{-\frac{1}{2}}\big(\mathbf{C}_i\big)\!\!+\!\!c_i\Big).
\end{align}
Note that the above upper bound can be much tighter than the one given in (\ref{eq:P1_upperbound}) since it utilizes the knowledge of $\{\mathbf{F}_j\}_{j\neq i}$. 

From the above discussion, we can utilize multi-block BCA method to solve the original problem ($\mathsf{P}0$). In each update, we tackle either one individual precoder (associated with one cluster) or the poster (associated with the FC). If the $i$-th cluster has $K_i>1$ sensors, then its precoder can be updated by the one-shot SDR-rank-reduction method; Otherwise, Theorem \ref{thm:closedform_P_i_D1} provides a clean closed-form solution. The entire approach is summarized in Algorithm \ref{alg:multiple_BCA}.  

\begin{algorithm}
\caption{Multi-Block RCA to solve ($\mathsf{P}0$)}
\label{alg:multiple_BCA}
\textbf{Initialization}: Randomly generate nonzero feasible $\{\mathbf{F}_i^{(0)}\}_{i=1}^{L}$ such that $\mathbf{g}^{(0)}$ obtained by (\ref{eq:opt_g}) is nonzero\;
\Repeat{the increase of SNR becomes sufficiently small or a predefined number of iterations is reached}
{
 \For{$i=1,\cdots,L;$}
 { 
      \eIf{$K_i>1$}
      {Solve ($\mathsf{P}7^i$); Then perform rank-reduction using Theorem 2.2 in \cite{bib:rank_one_decomp}; Then update $\mathbf{F}_i$ \;}
      {Set $bd_l$ as current $\mathsf{SNR}$; obtain $bd_u$ by (\ref{eq:bd_P_1_i})\;
	   \Repeat{$(bd_u-bd_l)$ is small enough}
	   {
	    Set $\alpha=(bd_u+bd_l)/2$\; 
	    Solve ($\mathsf{P}8_{\alpha}^i$) by Theorem (\ref{thm:closedform_P_i_D1})\; 
	    \eIf{$\mathsf{opt}(\mathsf{P}8_{\alpha}^i)\leq0$}
	    {$bd_l=\alpha$\;}
	    {$bd_u=\alpha$\;}
	   }
	   $\alpha=bd_l$\;
	   Solve ($\mathsf{P}8_{\alpha}^i$) by Theorem \ref{thm:closedform_P_i_D1}; update $\mathbf{F}_i$\; 
      }
      
  Update $\mathbf{g}$ by theorem \ref{thm:opt_g} \;
 }
}
\end{algorithm}

Although the multiple BCA method generates monotonically increasing SNR sequence, it is hard to prove that the limit points of its solution sequence guarantee to converge to stationary points of ($\mathsf{P}0$). Numerical results in section \ref{sec:numerical results} show that multiple BCA algorithm usually has a very satisfying convergence behaviors. 


By primal-dual inter point method \cite{bib:Lectures_BenTal}, the complexity of each outer-layer iteration of multiple BCA for homogeneous wireless sensor network is $\mathcal{O}\big(LK^{3.5}N^{3.5}\!+\!LM^3\big)$. Particularly for homogeneous network with $K=1$, the complexity becomes $\mathcal{O}\big(LM^3\big)$ with the help of theorem \ref{thm:closedform_P_i_D1}.

\section{Numerical Results}
\label{sec:numerical results}

In this section, numerical results are presented to testify the proposed algorithms' performance. In our experiments, the observation noise at each sensor is colored, which has a covariance 
\begin{align}
\mathbf{\Sigma}_i=\sigma_i^2\mathbf{\Sigma}_0,\ \ \ i\in\{1,\cdots,L\}, 
\end{align} 
where the $K_i\times K_i$ matrix $\mathbf{\Sigma}_0$ has the Toeplitz structure
\begin{align}
\label{eq:toeplitz_rho}
\mathbf{\Sigma}_{0}=
\left[
\begin{array}{cccc}
1 & \rho & \ddots & \rho^{K-1} \\
\rho & 1 & \ddots & \ddots \\
\ddots & \ddots &  \ddots & \rho \\
\rho^{K-1} & \ddots & \rho & 1
\end{array}
\right].
\end{align}

The parameter $\rho$ is set to $0.5$ for all sensors in the following experiments. Here we define the observation signal to noise ratio at the $i$-th sensor as $\mathsf{SNR}_i\triangleq\sigma_i^{-2}$ and the channel signal to noise ratio as $\mathsf{SNR}\triangleq\sigma_0^{-2}$. 

In figure \ref{fig:avg_SNR_SDR_vs_SOCP} and \ref{fig:avg_SNR_SDR_vs_L_BCA} the average SNR obtained at the FC are plotted. It is assumed that the sensor network has 5 sensors and FC has 4 antennas. We set $N_1=3, N_2=4, N_3=5, N_4=4, N_5=5$, $K_1=3,K_2=4,K_3=5,K_4=6,K_5=6$ and $P_1=0.2,P_2=0.2,P_3=0.3,P_4=0.2,P_5=0.3$. For each fixed channel SNR level, 50 random channel realizations are generated with each element of channel matrix follows circularly symmetric complex Gaussian distribution with zero mean and covariance 2. With channel SNR and channel matrices given, the proposed algorithms are performed starting from one common random initial. The obtained average SNR is plotted in figure \ref{fig:avg_SNR_SDR_vs_SOCP} and \ref{fig:avg_SNR_SDR_vs_L_BCA}. The obtained average SNR of SDR based and SOCP based 2BCA algorithms are plotted in figure \ref{fig:avg_SNR_SDR_vs_SOCP} with respect to different outer iterations. The curve associated with random initials actually represents the performance of random feasible linear transmitters. From figure \ref{fig:avg_SNR_SDR_vs_SOCP}, optimized SNR converges in 10 outer-iterations on average. These two algorithms have identical average convergence performance, this will also be verified by figure \ref{fig:Initials_SNR_SDR_vs_SOCP}. The average SNR performance obtained by multiple BCA algorithm is presented in figure  \ref{fig:avg_SNR_SDR_vs_L_BCA}, where SDR based 2BCA algorithm serves as a benchmark. Multiple BCA algorithm presents identical average SNR performance with the other 2 block algorithms.

\begin{figure}
\centering
\includegraphics[height=3.0in,width=3.8in]{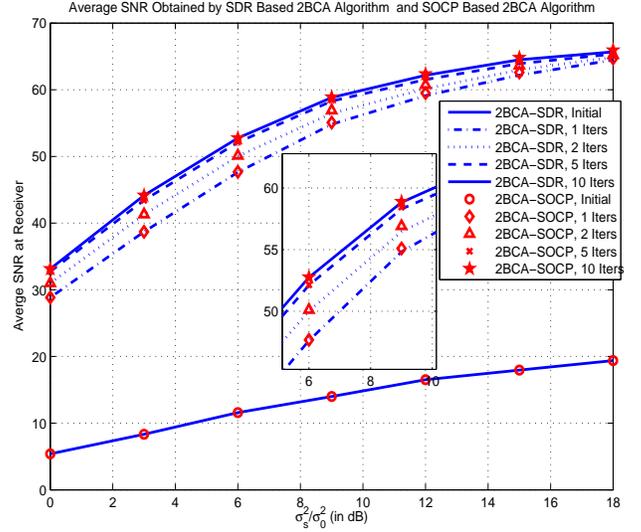}
\caption{Average SNR Obtained by SDR Based 2BCA Algorithm and SOCP Based 2BCA Algorithm}
\label{fig:avg_SNR_SDR_vs_SOCP} 
\end{figure}

\begin{figure}
\centering
\includegraphics[height=3.0in,width=3.8in]{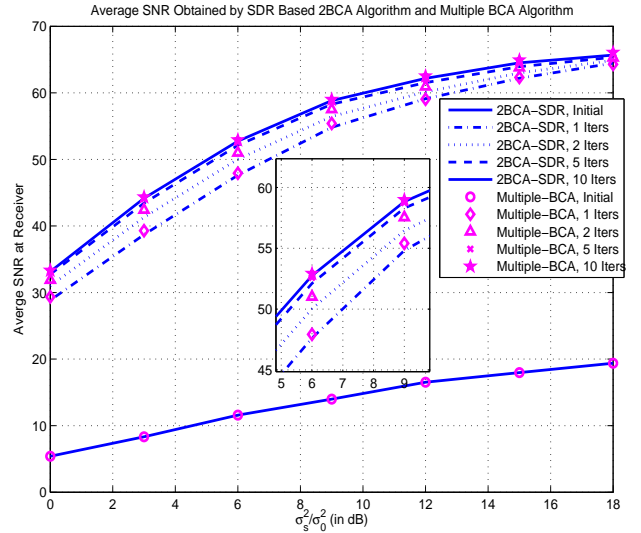}
\caption{Average SNR Obtained by SDP Based 2BCA Algorithm and Multiple BCA Algorithm}
\label{fig:avg_SNR_SDR_vs_L_BCA} 
\end{figure}

In figure \ref{fig:Initials_SNR_SDR_vs_SOCP} and \ref{fig:Initials_SNR_SDR_vs_L_BCA}, the impact of different initial points to the algorithms are examined. The system setup is identical with the experiment in figure \ref{fig:avg_SNR_SDR_vs_SOCP} and \ref{fig:avg_SNR_SDR_vs_L_BCA}. We set the channel SNR as 2dB and fix the channel matrices with one specific random realization. The three proposed algorithms are started from 10 different random initials and each SNR itinerary with respect to number of outer-layer iterations is plotted in figure \ref{fig:Initials_SNR_SDR_vs_SOCP} and  \ref{fig:Initials_SNR_SDR_vs_L_BCA}, where the itineraries of SDR based 2BCA algorithm serve as benchmarks. From figure \ref{fig:Initials_SNR_SDR_vs_SOCP} it can be seen that the two 2BCA algorithms have almost identical SNR itineraries. Comparatively, multiple BCA algorithm's itineraries are usually very different but finally it will converge to identical value. Figures \ref{fig:Initials_SNR_SDR_vs_SOCP} and \ref{fig:Initials_SNR_SDR_vs_L_BCA} reflect the fact that: the proposed three algorithms are initial-insensitive; they finally converge to identical SNR value; and usually 30 iterations are sufficient for these proposed algorithms to converge.

\begin{figure}
\centering
\includegraphics[height=3.0in,width=3.8in]{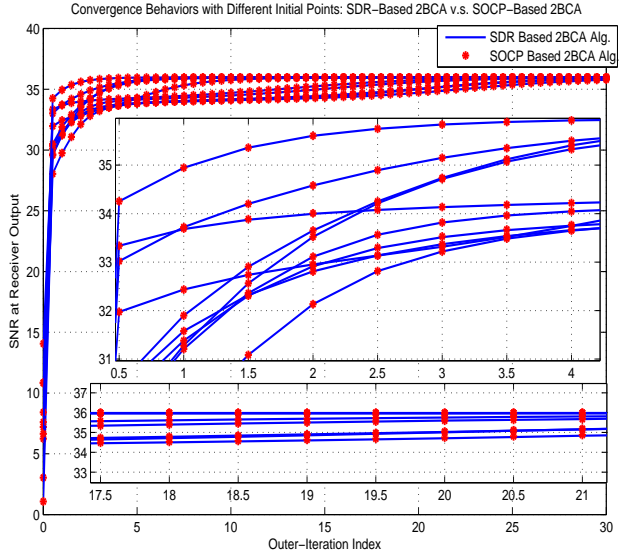}
\caption{Convergence with Different Initials: SDR Based 2BCA Algorithm v.s. SOCP Based 2BCA Algorithm}
\label{fig:Initials_SNR_SDR_vs_SOCP} 
\end{figure}

\begin{figure}
\centering
\includegraphics[height=3.0in,width=3.8in]{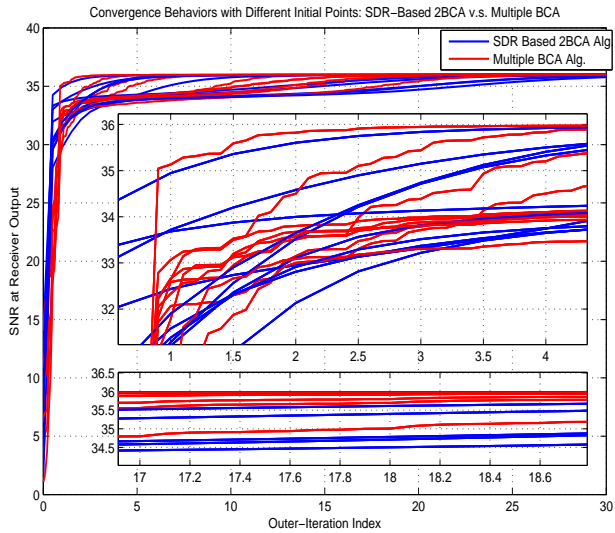}
\caption{Convergence with Different Initials: SDR Based 2BCA Algorithm v.s. Multiple BCA Algorithm}
\label{fig:Initials_SNR_SDR_vs_L_BCA} 
\end{figure}

\begin{figure}
\centering
\includegraphics[height=3.0in,width=3.8in]{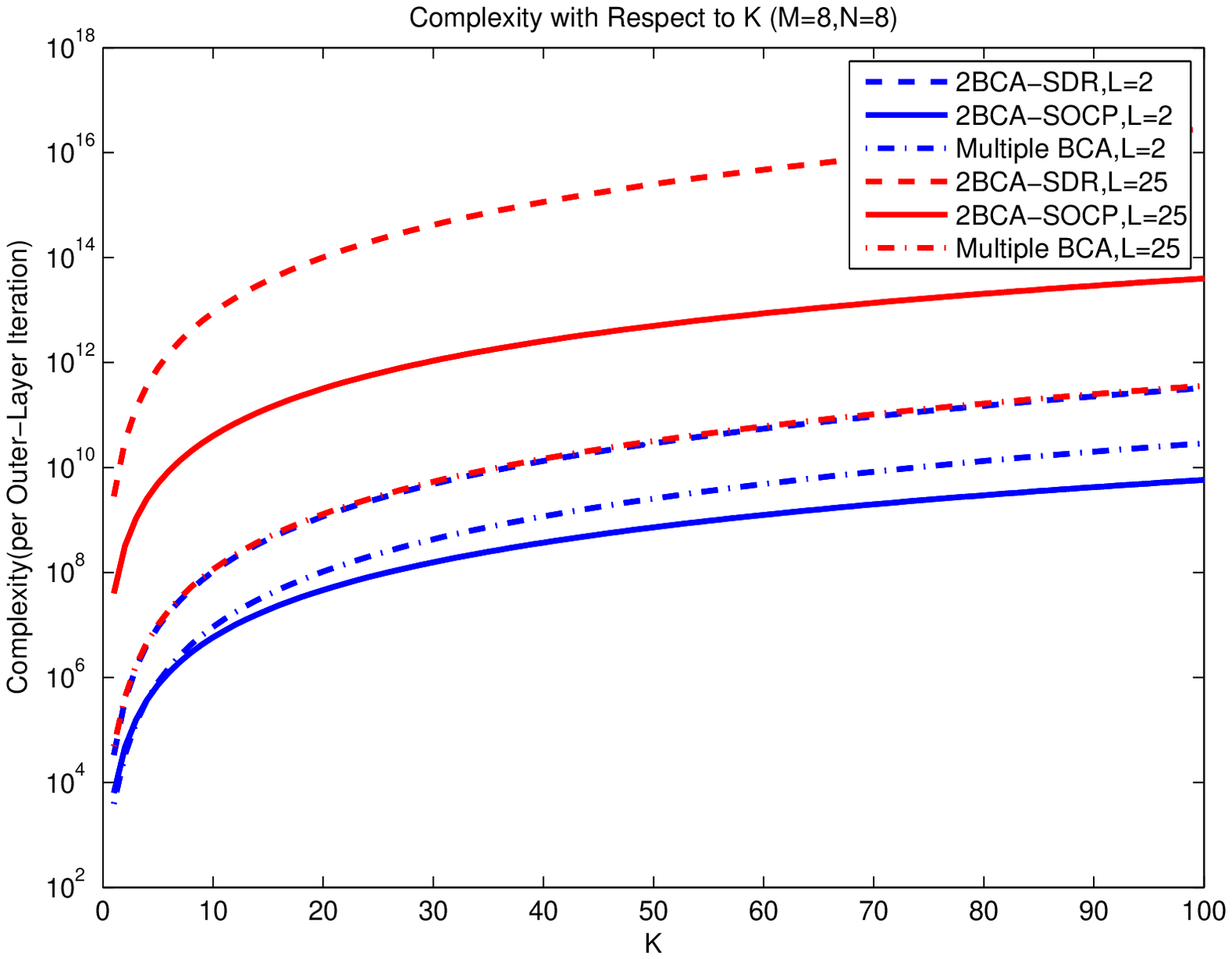}
\caption{Complexity of Algorithms with Respect to $K$}
\label{fig:complexity_K} 
\end{figure}

\begin{figure}
\centering
\includegraphics[height=3.0in,width=3.8in]{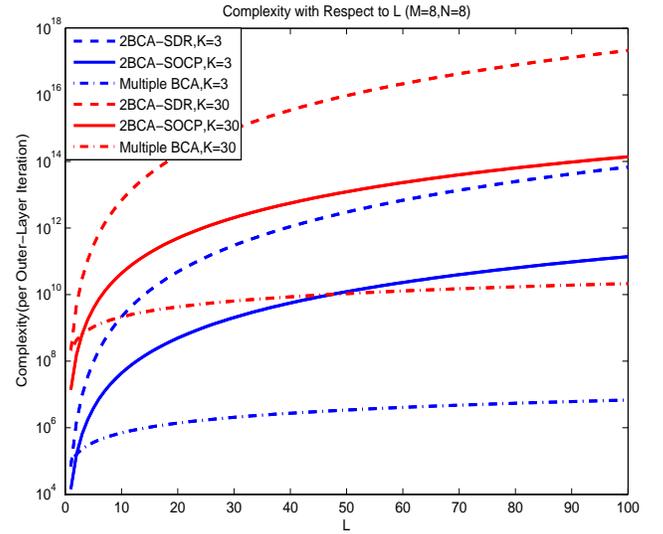}
\caption{Complexity of Algorithms with Respect to $L$}
\label{fig:complexity_L} 
\end{figure}

Next we present numerical results for complexity. Still we take homogeneous wireless sensor network as example. $N$ and $M$ denote the number of antennas for each sensor or sensor cluster and FC respectively and take modest values within several tens. Comparatively the number of sensors or sensor clusters can be large, and one cluster can have numerous sensors. So we focus on the impact of $L$ and $K$ on the complexity. Figure \ref{fig:complexity_K} and \ref{fig:complexity_L} represent the complexity for each outer-layer iteration for proposed algorithms with respect to $K$ and $L$ respectively. Generally SDR based 2BCA algorithm has higher complexity than the two others. The SOCP based algorithm has lowest complexity for large $K$ with small $L$ and multiple BCA algorithm has the lowest complexity for large $L$.

In the following the average execution time of proposed algorithms using MATLAB with the  standard toolbox CVX v2.1 on the same computer are presented in Table \ref{tab:runtime_1} and \ref{tab:runtime_2}. The multiple BCA algorithm requires much lower time for networks with large $L$ and SOCP based 2BCA algorithm is more efficient for large $K$ and small $L$. Although the complexity of SDR-based 2BCA algorithm increases drastically with the increase of $K$, $N$ and $L$ in general, it can still be useful in specific scenarios. Note that when the size of wireless sensor network is small, the execution time of SDR based 2BCA algorithm mainly comes from random samples generation and rescaling. In the case where parallel computation is available, this procedure can requires very little time and thus competitive to the other two algorithms.

\begin{table}
\caption{\small{MATLAB Running Time Per (Outer) Iteration(in sec.)}}
\centering
\begin{tabular}{|c|c|c|c|c|c|c|}
\hline
 Dim. & Alg. &$ L\!=\!5$ & $L\!=\!10$ & $L\!=\!20$ & $L\!=\!30$ & $L\!=\!40$ \\
\hline
$K\!=1\!$ & Alg.\ref{alg:2BCA_SDR} & 1.814 & 3.561 & 8.462 & 18.58 & 34.76\\
$N\!=\!4$ & Alg.\ref{alg:2BCA_SOCP} & 5.677 & 9.163 & 15.84 & 23.14 & 38.82 \\
 & Alg.\ref{alg:multiple_BCA} & 0.067 & 0.380 & 2.603 & 8.344 & 19.38 \\
\hline
$K\!=\!1$ & Alg.\ref{alg:2BCA_SDR} & 2.175 & 5.413 & 21.97 & 59.23 & 148.5 \\ 
$N\!=\!8$ & Alg.\ref{alg:2BCA_SOCP} & 7.488 & 12.21 & 25.58 & 51.19 & 84.17 \\ 
 & Alg.\ref{alg:multiple_BCA}  & 0.073 & 0.406 & 2.741 & 9.387 & 19.27 \\
\hline
$K\!=\!3$ & Alg.\ref{alg:2BCA_SDR} & 2.650 & 9.002 & 43.07 & 158.5 & 689.5 \\ 
$N\!=\!4$ & Alg.\ref{alg:2BCA_SOCP} & 10.462 & 21.40 & 54.79 & 111.9 & 45.60 \\ 
 & Alg.\ref{alg:multiple_BCA} & 1.106 & 2.423 & 6.549 & 14.29 & 26.95 \\
\hline
$K\!=\!3$ & Alg.\ref{alg:2BCA_SDR} & 7.222 & 32.32 & 536.9 & --- & --- \\ 
$N\!=\!8$ & Alg.\ref{alg:2BCA_SOCP} & 19.765 & 50.99 & 173.6 & 59.32 & 85.81 \\ 
 & Alg.\ref{alg:multiple_BCA} & 1.650 & 3.519 & 9.286 & 18.19 & 31.74 \\
\hline
$K\!=\!5$ & Alg.\ref{alg:2BCA_SDR} & 4.468 & 19.65 & 160.3 & --- & --- \\ 
$K\!=\!4$ & Alg.\ref{alg:2BCA_SOCP} & 14.944 & 32.85 & 125.8 & 50.04 & 69.08  \\ 
 & Alg.\ref{alg:multiple_BCA} & 1.455 & 2.989 & 7.749 & 16.63 & 30.41 \\
\hline
$K\!=\!5$ & Alg.\ref{alg:2BCA_SDR} & 16.442 & 115.8 & --- & --- & --- \\ 
$N\!=\!8$ & Alg.\ref{alg:2BCA_SOCP} & 32.273 & 121.0 & 80.21 & 134.2 & 201.3 \\ 
 & Alg.\ref{alg:multiple_BCA} & 2.662 & 5.617 & 13.51 & 24.99 & 42.33 \\
\hline
\end{tabular}
\label{tab:runtime_1}
\\ Note: ``---'' means the problem is too large to be solved. \\
Alg.1: SDR-2BCA alg.; 
Alg.2: SOCP-2BCA alg.;
Alg.3: multiple BCA alg.  
\end{table}

\begin{table}
\caption{\small{MATLAB Running Time Per (Outer) Iteration(in sec.)}}
\centering
\begin{tabular}{|c|c|c|c|c|c|}
\hline
 Dim. & Alg. &$ K\!=\!20$ & $K\!=\!30$ & $K\!=\!40$ \\
\hline
$L\!=2\!$ & Alg.\ref{alg:2BCA_SDR} & --- & --- & --- \\
$N\!=\!16$ & Alg.\ref{alg:2BCA_SOCP} & 625.8 & $1.903\times10^3$ & $5.378\!\times\!10^3$ \\
$M=3$ & Alg.\ref{alg:multiple_BCA} & 89.47 & $2.171\!\times\!10^3$ & --- \\
\hline
\end{tabular}
\label{tab:runtime_2}
\\ Note: ``---'' means the problem is too large to be solved. \\
Alg.1: SDR-2BCA alg.; 
Alg.2: SOCP-2BCA alg.;
Alg.3: multiple BCA alg.    
\end{table}

\section{Conclusion}
\label{sec:conclusion}

This paper considers the joint transceiver design problem in cluster based wireless sensor network. To maximize the output SNR at the fusion center, the difficult original problem is decomposed into two or more subproblems and solution to each subproblem is obtained. Convergence and complexity are carefully examined. Extensive numerical results show that the proposed algorithms provide equivalently good SNR values while have different efficiency characteristics and  suitable for various system settings. As an extension of current problem, robust design and decentralized algorithms are desirable and meaningful for future study.

\appendix

\subsection{Proof of Lemma \ref{lem:P2_P3_equivalent}}
\label{subsec:appendix_lem_P2_P3_equivalent}
\begin{proof}
Assume that $\mathbf{X}^{\star}$ and $(\mathbf{Y}^{\star},\nu^{\star})$ are optimal solutions to ($\mathsf{P}2$) and ($\mathsf{P3}$) respectively, and $\mathsf{opt}(\mathsf{P}2)$ and $\mathsf{opt}(\mathsf{P}3)$ are optimal values of the two problems. 

First we claim that $\nu^{\star}>0$. This can be proved by contradiction. If $\nu^{\star}=0$, then we readily obtain $\Tr\big\{\mathbf{D}_i\mathbf{Y}^{\star}\big\}=0$, for $i=1,\cdots,L$. This leads to $\Tr\big\{(\sum_{i=1}^{L}\mathbf{D}_i)\mathbf{Y}^{\star}\big\}=0$. Since it is assumed that $\mathbf{\Sigma}_i\succ0$, for $i\in\{1,\cdots,L\}$, it holds that $\mathbf{C}_i\succ0$, for $i\in\{1,\cdots,L\}$. Thus $\sum_{i=1}^{L}\mathbf{D}_i=\mathsf{Diag}\{\mathbf{C}_1,\cdots,\mathbf{C}_L\}\succ0$ and we obtain $\mathbf{Y}^{\star}=\mathbf{O}$.
However this violates the constraint (\ref{eq:P3_constr1}), since its left hand side equals zero. Thus $\nu^{\star}>0$.

If $(\mathbf{Y}^{\star},\nu^{\star})$ solves ($\mathsf{P}3$), since $\nu^{\star}>0$, it is easy to check $\mathbf{Y}^{\star}/\nu^{\star}$ is feasible for ($\mathsf{P}2$) and gives an objective value of $\frac{\Tr\{\mathbf{A}(\mathbf{Y}^{\star}/\nu^{\star})\}}{\Tr\{\mathbf{B}(\mathbf{Y}^{\star}/\nu^{\star})+c_0\}}=\mathsf{opt}(\mathsf{P}3)$. So $\mathsf{opt}(\mathsf{P}3)\leq\mathsf{opt}(\mathsf{P}2)$. On the other hand, if $\mathbf{X}^{\star}$ solves ($\mathsf{P}2$), then $\left(\frac{\mathbf{X}^{\star}}{\Tr\{\mathbf{B}\mathbf{X}^{\star}+c_0\}},\frac{1}{\Tr\{\mathbf{B}\mathbf{X}^{\star}+c_0\}}\right)$ is a feasible solution to ($\mathsf{P}3$) and gives objective value of $\Tr\{\mathbf{A}\frac{\mathbf{X}^{\star}}{\Tr\{\mathbf{B}\mathbf{X}^{\star}+c_0\}}\}=\mathsf{opt}(\mathsf{P}2)$. So $\mathsf{opt}(\mathsf{P}2)\leq\mathsf{opt}(\mathsf{P}3)$. The proof is complete.
\end{proof}

\subsection{Proof of Lemma \ref{lem:P3_solvability}}
\label{subsec:appendix_lem_P3_solvability}
\begin{proof}
First we prove that ($\mathsf{P}3$) is solvable. By (\ref{eq:P3_constr1}) we have $0\leq\nu\leq c_0^{-1}$, so $\nu$ is bounded. Combining (\ref{eq:P3_constr2}) we readily obtain $\Tr\{\mathbf{D}_i\mathbf{Y}\}\leq P_i\nu\leq P_i/c_0$, $i\in\{1,\cdots,L\}$, which implies $\Tr\{(\sum_{i=1}^{L}\mathbf{D}_i)\mathbf{Y}\}\leq(\sum_{i=1}^{L}P_i)/c_0$. Since $(\sum_{i=1}^{L}\mathbf{D}_i)=\mathsf{Diag}\{\mathbf{C}_1,\cdots,\mathbf{C}_L\}\succ0$, this means $\mathbf{Y}$ is bounded. So the feasible region of ($\mathsf{P}3$) is bounded. Obviously the feasible region of $(\mathbf{Y},\nu)$ is also closed. So ($\mathsf{P}3$) has compact feasible region. Since the objective $\Tr\{\mathbf{A}\mathbf{Y}\}$ always takes finite values on the whole feasible region, by Weierstrass' theorem(proposition 3.2.1-(1) in \cite{bib:convex_optimization_theory}), ($\mathsf{P}3$) is solvable.

The Lagrangian function of problem ($\mathsf{P}3$) is given as
\begin{align}
\!\!\!\!&\qquad\mathcal{L}\big(\mathbf{Y},\nu,\lambda,\{\mu_i\}_{i=1}^{L}\big)\\
\!\!\!\!&=\!\Tr\big\{\mathbf{Y}\mathbf{A}\big\}\!+\!\lambda\Big(1\!\!-\!\!\Tr\big\{\mathbf{Y}\mathbf{B}\big\}\!\!-\!\!c_0\nu\Big)\!-\!\sum_{i=1}^{L}\mu_i\Big(\Tr\big\{\!\mathbf{Y}\mathbf{D}_i\big\}\!\!-\!\!P_i\nu\Big)\nonumber\\
\!\!\!\!&=\Tr\Big\{\Big[\mathbf{A}\!-\!\lambda\mathbf{B}\!-\!\sum_{i=1}^L\mu_i\mathbf{D}_i\Big]\mathbf{Y}\Big\}\!+\!\Big(-c_0\lambda\!+\!\sum_{i=1}^{L}P_i\mu_i\Big)\nu\!+\!\lambda.\nonumber
\end{align}
By taking the supremum of Lagrangian function with respect to $\mathbf{Y}\succcurlyeq0$ and $\nu\geq0$, the dual function is obtained as
\begin{align}
g(\lambda,\{\mu_i\}_{i=1}^L)=\underset{\mathbf{Y}\succcurlyeq0,\nu\geq0}{\sup.}\mathcal{L}\big(\mathbf{Y},\nu,\lambda,\{\mu_i\}_{i=1}^{L}\big)=\lambda
\end{align}
with the conditions $\big[\mathbf{A}-\lambda\mathbf{B}-\sum_{i=1}^L\mu_i\mathbf{D}_i\big]\preccurlyeq0$ and $\big(-c_0\lambda+\sum_{i=1}^{L}P_i\mu_i\big)\leq0$ satisfied. So the dual problem of ($\mathsf{P}3$) can be given as 
\begin{subequations}
\begin{align}
\!\!\!\!\!\!\!\!\!\!\!\!\!\!\!\!\!\!\!\!(\mathsf{D}3): \underset{\lambda,\{\mu_i\}_{i=1}^L}{\min.}\ & \lambda \\
\!\!\!\!\!\!\!\!\!\!\!\!\!\!\!\!\!\!\!\!\mathsf{s.t.}\ &\lambda\mathbf{B}+\sum_{i=1}^{L}\mu_i\mathbf{D}_i\succcurlyeq\mathbf{A},\label{eq:D3_constr1}\\
\!\!\!\!\!\!\!\!\!\!\!\!\!\!\!\!\!\!\!\!& c_0\lambda\geq\sum_{i=1}^{L}P_i\mu_i, \label{eq:D3_constr2}\\
\!\!\!\!\!\!\!\!\!\!\!\!\!\!\!\!\!\!\!\!& \mu_i\geq0, \ i=1,\cdots,L.\label{eq:D3_constr3}
\end{align}
\end{subequations}
Next we prove that ($\mathsf{D}3$) is solvable. To do this it is sufficient to show that there exists a real value $\gamma$ such that the level set $\{(\lambda,\{\mu_i\}_{i=1}^{L})|\lambda\leq\gamma, (\lambda,\{\mu_i\}_{i=1}^{L})\in\mathsf{dom}(\mathsf{D}3)\}$ is nonempty and bounded, where $\mathsf{dom}(\mathsf{D}3)$ means feasible region of ($\mathsf{D}3$). Here we choose $\tilde{\mu}_i=\lambda_{max}(\mathbf{A})/\lambda_{min}(\mathbf{C}_i)$ for $i=1,\cdots,L$, where $\lambda_{max}(\cdot)$ and $\lambda_{max}(\cdot)$ represent the maximal and minimal eigenvalue of a matrix respectively. Set $\tilde{\lambda}=c_0^{-1}(\sum_{i=1}^L\tilde{\mu}_iP_i)$. By definition the constraints (\ref{eq:D3_constr2}) and (\ref{eq:D3_constr3}) are satisfied by $(\tilde{\lambda},\{\tilde{\mu}_i\}_{i=1}^{L})$. Since $\mathbf{B}\succcurlyeq0$ and $\tilde{\lambda}\geq0$,
\begin{align}
\tilde{\lambda}\mathbf{B}+\sum_{i=1}^{L}\tilde{\mu}_i\mathbf{D}_i&\succcurlyeq\sum_{i=1}^{L}\tilde{\mu}_i\mathbf{D}_i=\mathsf{Diag}\{\tilde{\mu}_1\mathbf{C}_1,\cdots,\tilde{\mu}_L\mathbf{C}_L \}\nonumber\\
&\succcurlyeq\lambda_{\max}\big({\mathbf{A}}\big)\mathbf{I}_{\sum_{i=1}^{L}K_iN_i}\succcurlyeq\mathbf{A}.
\end{align}
Thus constraint (\ref{eq:D3_constr1}) is also satisfied by $(\tilde{\lambda},\{\tilde{\mu}_i\}_{i=1}^{L})$. Set $\tilde{\gamma}=\tilde{\lambda}$. Combination of $\lambda\leq\tilde{\gamma}$ and the constraint (\ref{eq:D3_constr2}) guarantees that $\lambda$ and all $\mu_i$'s are bounded. So we conclude that the level set $\{(\lambda,\{\mu_i\}_{i=1}^{L})|\lambda\leq\tilde{\gamma},(\lambda,\{\mu_i\}_{i=1}^{L})\in\mathsf{dom}(\mathsf{D}3)\}$ is nonempty and bounded. Invoking  Weierstrass' theorem(proposition 3.2.1-(2) in \cite{bib:convex_optimization_theory}), ($\mathsf{D}3$) is solvable.
\end{proof}

\subsection{Proof of Lemma \ref{lem:A_rank1}}
\label{subsec:appendix_lem_A_rank1}

\begin{proof}
Recalling the definition of $\mathbf{A}_{ij}$ in (\ref{eq:definition_fABCc_A}) and utilizing the identity $\big(\mathbf{A}\mathbf{B}\big)\otimes\big(\mathbf{C}\mathbf{D}\big)=\big(\mathbf{A}\otimes\mathbf{C}\big)\big(\mathbf{B}\otimes\mathbf{D}\big)$ \cite{bib:complex_matrix_derivative}, we have
\begin{align}
\mathbf{A}_{ij}&=\big(\mathbf{1}_{K_i}\mathbf{1}_{K_j}^T\big)\otimes\big(\mathbf{H}_i^H\mathbf{g}\mathbf{g}^H\mathbf{H}_j\big)\nonumber\\
&=\big(\mathbf{1}_{K_i}\otimes\mathbf{H}_i^H\mathbf{g}\big)\big(\mathbf{1}^T_{K_j}\otimes\mathbf{g}^H\mathbf{H}_j\big).
\end{align}
Then the $j$-th block column of $\mathbf{A}$ is given as 
\begin{align}
\mathbf{A}_{:j}&\!=\!\left[\!
\begin{array}{c}
\mathbf{A}_{1j} \\
\vdots \\
\mathbf{A}_{Lj}
\end{array}
\!\right]
\!=\!\left[\!
\begin{array}{c}
\big(\mathbf{1}_{K_1}\otimes\mathbf{H}_1^H\mathbf{g}\big)\big(\mathbf{1}^T_{K_j}\otimes\mathbf{g}^H\mathbf{H}_j\big)\\
\vdots\\
\big(\mathbf{1}_{K_L}\otimes\mathbf{H}_L^H\mathbf{g}\big)\big(\mathbf{1}^T_{K_j}\otimes\mathbf{g}^H\mathbf{H}_j\big)
\end{array}
\!\right]\\
\!\!\!\!\!\!&\!=\!\left[\!
\begin{array}{c}
\big(\mathbf{1}_{K_1}\otimes\mathbf{H}_1^H\mathbf{g}\big)\\
\vdots\\
\big(\mathbf{1}_{K_L}\otimes\mathbf{H}_L^H\mathbf{g}\big)
\end{array}
\!\right]\big(\mathbf{1}^T_{K_j}\otimes\mathbf{g}^H\mathbf{H}_j\big)\\
&\!=\!\mathbf{a}\big(\mathbf{1}^T_{K_j}\otimes\mathbf{g}^H\mathbf{H}_j\big).
\end{align} 
The last equality utilizes the definition of $\mathbf{a}$ in (\ref{eq:def_a}). Then the  matrix $\mathbf{A}$ can be represented by packing all its column blocks as follows
\begin{align}
\mathbf{A}
&=\left[
\mathbf{A}_{:1},\cdots,\mathbf{A}_{:L}
\right]\\
&=\Big[
\mathbf{a}\big(\mathbf{1}_{K_1}^T\otimes\mathbf{g}^H\mathbf{H}_1\big),\cdots,\mathbf{a}\big(\mathbf{1}_{K_L}^T\otimes\mathbf{g}^H\mathbf{H}_L\big)
\Big]\\
&=\mathbf{a}\Big[\big(\mathbf{1}_{K_1}^T\otimes\mathbf{g}^H\mathbf{H}_1\big),\cdots,\big(\mathbf{1}_{K_L}^T\otimes\mathbf{g}^H\mathbf{H}_L\big)\Big]\\
&=\mathbf{a}\mathbf{a}^H.
\end{align}
The proof is complete.
\end{proof}

\subsection{Proof of Lemma \ref{lem:opt_P1_upperbound}}
\label{subsec:appendix_lem_opt_P1_upperbound}

\begin{proof}
By the $i$-th power constraint (\ref{eq:opt_prob_P1_constr2}) we have
\begin{align}
\label{eq:bd_norm_f_i}
\lambda_{min}\big(\mathbf{C}_i\big)\|\mathbf{f}_i\|_2^2\leq\mathbf{f}_i^H\mathbf{C}_i\mathbf{f}_i\leq{P}_i,
\end{align}
which implies 
\begin{align}
\label{eq:F_i_bound}
\|\mathbf{f}_i\|_2\leq\sqrt{\frac{P_i}{\lambda_{min}(\mathbf{C}_i)}}, \ i=1,\cdots,L.
\end{align}
By Cauchy-Schwarz inequality the numerator $\mathbf{f}^H\mathbf{A}\mathbf{f}$ of SNR is bounded as
\begin{align}
\mathbf{f}^H\mathbf{A}\mathbf{f}&=\big|\mathbf{a}^H\mathbf{f}\big|^2\leq\Big|\sum_{i=1}^{L}\big|\mathbf{f}_i^H\big(\mathsf{1}_{K_i}\otimes\mathbf{H}_i^H\mathbf{g}\big)\big|\Big|^2\\
&\leq\Big|\sum_{i=1}^L\big\|\mathbf{f}_i\big\|_2\big\|\mathbf{1}_{K_i}\otimes\mathbf{H}_i^H\mathbf{g}\big\|_2\Big|^2\\
&\leq\left(\sum_{i=1}^L\sqrt{\frac{P_i}{\lambda_{min}\big(\mathbf{C}_i\big)}}K_i\big\|\mathbf{H}_i^H\mathbf{g}\big\|_2\right)^2,
\end{align}
where the above first inequality uses Lemma \ref{lem:A_rank1}. Combining the fact that $\mathbf{f}^H\mathbf{B}\mathbf{f}+c_0\geq c_0$, the upper bound in the lemma is proved.
\end{proof}

\subsection{Proof of Theorem \ref{thm:2BCA_convergence}}
\label{subsec:appendix_thm_2BCA_convergence}

\begin{proof}
Since each update of $\{\mathbf{F}_i\}_{i=1}^{L}$ or $\mathbf{g}$ is obtained by solving a maximization problem, SNR monotonically increases. 
At the time we note that SNR is bounded. In fact since the SNR is invariant to scaling of $\mathbf{g}$ we can assume that $\|\mathbf{g}\|_2=1$. According to (\ref{eq:F_i_bound}) in the proof of lemma \ref{lem:opt_P1_upperbound}, $\mathbf{F}_i$ is bounded for all $i=1,\cdots,L$. Thus the numerator of SNR is bounded above and the denominator of SNR is bounded away from zero, so SNR should be bounded. Consequently the objective value sequence by algorithms \ref{alg:2BCA_SDR} or \ref{alg:2BCA_SOCP} converges since it is monotonically increasing and bounded.

Since $\{\mathbf{F}_i\}_{i=1}^{L}$ are bounded, by Bolzano-Weierstrass theorem \cite{bib:analysis_rudin} there exists a sequence $\{j_k\}_{k=1}^{\infty}$ such that $\big\{\{\mathbf{F}_i^{(j_k)}\}_{i=1}^{L}\big\}_{k=1}^{\infty}$ converges. Since $\mathbf{g}^{(j_k)}$ is updated by (\ref{eq:opt_g}) which is a continuous function of $\{\mathbf{F}_i^{(j_k)}\}_{i=1}^{L}$, thus the sequence $\big\{\big(\{\mathbf{F}_i^{(j_k)}\}_{i=1}^{L},\mathbf{g}^{(j_k)}\big)\big\}_{k=1}^{\infty}$ also converges. The existence of limit points of the solution sequence is proved.

The feasible region of ($\mathsf{P}0$) is a Cartesian product $\mathcal{X}_1\times\mathcal{X}_2$ with  $\mathcal{X}_1\triangleq\big\{\{\mathbf{F}_i\}_{i=1}^L\big|\text{(\ref{eq:original_SNR_prob_constr}) is satisfied for }i=1,\cdots,L\big\}$ and $\mathcal{X}_2\triangleq\mathbb{C}^{M\times1}\backslash\{\mathbf{0}\}$. Corollary 2 in \cite{bib:Gauss-Seidel} states that any limit point of solution sequence generated by 2-block coordinate ascent method is stationary. It should be noted that this conclusion is obtained under the assumption that the objective function is continuously differentialbe on feasible region and each block feasible region(each term in the Cartesian product) is nonempty, closed and convex set.  Unfortunately the problem ($\mathsf{P}0$) does not satisfy this assumption since 
$\mathcal{X}_2$ is nonconvex and not closed. In the following we will show that conclusion in \cite{bib:Gauss-Seidel} still applies to our problem after appropriately adjusting its argument.

First we assert that the solution sequence always has nonzero $\mathbf{g}$, i.e. $\mathbf{g}^{(k)}\neq0$ for all $k=0,1,\cdots$. Since algorithms \ref{alg:2BCA_SDR} or \ref{alg:2BCA_SOCP} starts from $\big(\{\mathbf{F}_i^{(0)}\}_{i=1}^{L},\mathbf{g}^{(0)}\big)$ with $\mathbf{g}^{(0)}\neq\mathbf{0}$, the assertion holds for $k=0$ and $\mathsf{SNR}\big(\{\mathbf{F}_i^{(0)}\}_{i=1}^{L},\mathbf{g}^{(0)}\big)>0$. Assume that $m\geq1$ is the smallest integer such that $\mathbf{g}^{(m)}=\mathbf{0}$, then according to (\ref{eq:opt_g}) $\big(\sum_{i=1}^L\mathbf{H}_i\mathbf{F}_i^{(m)}\mathbf{1}_{K_i}\big)=\mathbf{0}$. Notice $\mathbf{g}^{(m-1)}\neq\mathbf{0}$. By (\ref{eq:SNR_func}) this readily implies $\mathsf{SNR}\big(\{\mathbf{F}_i^{(m)}\}_{i=1}^{L},\mathbf{g}^{(m-1)}\big)=0<\mathsf{SNR}\big(\{\mathbf{F}_i^{(0)}\}_{i=1}^{L},\mathbf{g}^{(0)}\big)$, which contradicts the increasing monotonicity of SNR. 

Next we assert that any limit point $(\{\bar{\mathbf{F}}_i\}_{i=1}^{L},\bar{\mathbf{g}})$ of solution sequence has nonzero $\bar{\mathbf{g}}$. By contradiction we assume that the subsequence $\big\{\big(\{\mathbf{F}_i^{(j_k)}\}_{i=1}^{L},\mathbf{g}^{(j_k)}\big)\big\}_{k=1}^{\infty}$ converges to $(\{\bar{\mathbf{F}}_i\}_{i=1}^{L},\bar{\mathbf{g}})$ with $\bar{\mathbf{g}}=\mathbf{0}$. Then by (\ref{eq:opt_g}) $\big\{\sum_{i}^{L}\mathbf{H}_i\mathbf{F}_i^{(j_k)}\mathbf{1}_{K_i}\big\}_{k=1}^{\infty}\rightarrow\mathbf{0}$. By rescaling each $\mathbf{g}^{(j_k)}$ to $\hat{\mathbf{g}}^{(j_k)}$ such that $\|\hat{\mathbf{g}}^{(j_k)}\|_2=1$ for all $k=1,2,\cdots$, we actually construct another solution sequence which is also generated by 2-block coordinate ascent method, since scaling of $\mathbf{g}$ does not change the SNR value. Now for this new solution sequence $\big\{\big(\{\mathbf{F}_i^{(j_k)}\}_{i=1}^{L},\hat{\mathbf{g}}^{(j_k)}\big)\big\}_{k=1}^{\infty}$, since $\big\{\sum_{i}^{L}\mathbf{H}_i\mathbf{F}_i^{(j_k)}\mathbf{1}_{K_i}\big\}_{k=1}^{\infty}\rightarrow\mathbf{0}$ while $\{\hat{\mathbf{g}}^{(j_k)}\}_{k=1}^{\infty}$(consequently the denominator of SNR) is bounded away from zero, we have $\mathsf{SNR}\big(\{\mathbf{F}_i^{(j_k)}\}_{i=1}^{L},\hat{\mathbf{g}}^{(j_k)} \big)\rightarrow0$, which again contradicts the increasing monotonicity of SNR sequence.

In \cite{bib:Gauss-Seidel} the closedness assumption of $\mathcal{X}_2$ is implicitly invoked in its proposition 2 to ensure that any limit point of solution sequence is feasible. Through the above proof we can see that this result holds true thus proposition 2 in \cite{bib:Gauss-Seidel} applies to our problem. 

The convexity assumption of $\mathcal{X}_2$ is explicitly utilized in \cite{bib:Gauss-Seidel} in its proof of proposition 3. Here we identify the notations $i$, $\mathbf{x}_{i+1}$, $\mathcal{X}_{i+1}$ and $\mathbf{w}(k,i)$ used in the original proof of proposition 3 in \cite{bib:Gauss-Seidel} as $1$, $\mathbf{g}$, $\mathcal{X}_2$ and $\big(\{\mathbf{F}_i^{(k+1)}\}_{i=1}^L,\mathbf{g}^{(k)}\big)$ respectively in our case.

According to the proof in \cite{bib:Gauss-Seidel}, we can find a descent direction $\mathbf{d}_{2}^{(k)}=\tilde{\mathbf{g}}-\mathbf{g}^{(k)}$ with $\tilde{\mathbf{g}}\in\mathcal{X}_2$ and
\begin{align}
\label{eq:conv_proof_contrad_2}
\mathbf{\nabla}_{\mathbf{g}^{(k)}}\mathsf{SNR}\big(\{\mathbf{F}_i^{(k+1)}\}_{i=1}^L,\mathbf{g}^{(k)}\big)^T\big(\tilde{\mathbf{g}}-\mathbf{g}^{(k)}\big)<0. 
\end{align}
 ((\ref{eq:conv_proof_contrad_2}) corresponds to the inequality $\nabla_{i+1}f(\mathbf{w}(k,i))^Td_{i+1}^k<0$ in the original proof of proposition 3 in \cite{bib:Gauss-Seidel}, which lies under equation (11) and is not labeled with number). Then by Armijo-type line search we can update $\mathbf{g}^{(k)}$ with $\mathbf{g}^{(k)}+\alpha_2^{(k)}\mathbf{d}_{2}^{(k)}$, where $\mathbf{d}^{(k)}_2\in\mathcal{X}_2$ and $\alpha_2^{(k)}\in(0,1]$. The convexity of $\mathcal{X}_2$ in \cite{bib:Gauss-Seidel} guarantees that $\mathbf{g}^{(k)}+\alpha_2^{(k)}\mathbf{d}_{2}^{(k)}\in\mathcal{X}_2$. 

Now we show that the fact $\mathbf{g}^{(k)}+\alpha_2^{(k)}\mathbf{d}_{2}^{(k)}\in\mathcal{X}_2$ still holds for our problem although our $\mathcal{X}_2$ is nonconvex. By contradiction assume that $\mathbf{g}^{(k)}+\alpha_2^{(k)}\mathbf{d}_{2}^{(k)}\notin\mathcal{X}_2$, i.e. 
\begin{align}
\label{eq:conv_proof_contrad_1}
\mathbf{g}^{(k)}+\alpha_2^{(k)}\mathbf{d}_{2}^{(k)}=\mathbf{0}. 
\end{align}
This is actually impossible. By substituting $\mathbf{d}_{2}^{(k)}=\tilde{\mathbf{g}}-\mathbf{g}^{(k)}$ into (\ref{eq:conv_proof_contrad_1}) we have
\begin{align}
(\alpha_{2}^{(k)}-1)\mathbf{g}^{(k)}=\alpha_{2}^{(k)}\tilde{\mathbf{g}}.
\end{align}
As a result of Armijo-type line search algorithm(refer to (3) and proposition 1 in \cite{bib:Gauss-Seidel}), $(\alpha_{2}^{(k)}-1)\in[0,1)$. If $\alpha_{2}^{(k)}=1$, then $\tilde{\mathbf{g}}=\mathbf{0}$, which contradicts the fact $\tilde{\mathbf{g}}\in\mathcal{X}_2$. If $(\alpha_{2}^{(k)}-1)<1$, then $\tilde{\mathbf{g}}=\frac{(\alpha_{2}^{(k)}-1)}{\alpha_{2}^{(k)}}\mathbf{g}^{(k)}$. This is also impossible since $\mathsf{SNR}$ is invariant to scaling of $\mathbf{g}$ and thus $\mathbf{\nabla}_{\mathbf{g}^{(k)}}\mathsf{SNR}\big(\{\mathbf{F}_i^{(k+1)}\}_{i=1}^L,\mathbf{g}^{(k)}\big)^T\big(\tilde{\mathbf{g}}-\mathbf{g}^{(k)}\big)=0$, which contradicts the fact (\ref{eq:conv_proof_contrad_2}). Thus the proposition 3 in \cite{bib:Gauss-Seidel} also stands for our problem. 

As a direct implication of proposition 2 and 3, the corollary 2 in \cite{bib:Gauss-Seidel}  holds true and thus any limit point provided by algorithm \ref{alg:2BCA_SDR} or \ref{alg:2BCA_SOCP} is stationary point of ($\mathsf{P}0$).
\end{proof}

\subsection{Proof of Theorem \ref{thm:solvability_S_pros}}
\label{subsec:appendix_thm_solvability_S_pros}

\begin{proof} 
Since the problem ($\mathsf{P}8_{\alpha}^i$) is a quadratic problem with one quadratic constraint and is obviously strictly feasible, the result of Appendix B.1 in \cite{bib:CvxOpt} is valid to invoke, which states that ($\mathsf{P}8_{\alpha}^i$) has the following relaxation 
\begin{subequations}
\begin{align}
\!\!\!\!\!\!\!(\mathsf{P}9_{\alpha}^i)\underset{\mathbf{X},\mathbf{x}}{\min.}\ &\Tr\big\{\!\big[\alpha\mathbf{B}_i\!\!-\!\!\mathbf{A}_{ii}\!\big]\!\mathbf{X}\!\big\}\!\!-\!\!2\mathsf{Re}\big\{\!\mathbf{q}_i^H\mathbf{x}\big\}\!\!+\!\!(\alpha d_i\!\!-\!\!c_i\!), \\
\mathrm{s.t.}\ &\Tr\big\{\mathbf{C}_i\mathbf{X}\big\}-P_i\leq0, \\
&\left[
\begin{array}{cc}
\mathbf{X} & \mathbf{x} \\
\mathbf{x}^H & 1
\end{array}
\right]\succeq 0.
\end{align}
\end{subequations}
with $\mathsf{opt}(\mathsf{P}9_{\alpha}^i)\!=\!\mathsf{opt}(\mathsf{P}8_{\alpha}^i)$. We replace the variables $(\mathbf{X},\mathbf{x})$ in $(\mathsf{P}9^i_{\alpha})$ by one matrix variable $\widetilde{\mathbf{X}}$ and rewrite it into a SDP form
\begin{subequations}
\begin{align}
(\mathsf{P}10_{\alpha}^i)\ \underset{\widetilde{\mathbf{X}}}{\min.}\ &\Tr\big\{\mathbf{P}_1\widetilde{\mathbf{X}}\big\},\\ 
\mathsf{s.t.}\ &\Tr\big\{\mathbf{P}_2\widetilde{\mathbf{X}}\big\}\leq P_i, \\
&\Tr\big\{\mathbf{P}_3\widetilde{\mathbf{X}}\big\}=1,
\end{align}
\end{subequations}
with the parameter matrices being defined as 
\begin{subequations}
\begin{align}
\!\!\!\!\mathbf{P}_1\!\!\triangleq\!\!&
\left[\!\!
\begin{array}{cc}
\alpha\mathbf{B}_i\!\!-\!\!\mathbf{A}_{ii} & -\mathbf{q}_i\! \\
-\mathbf{q}_i^H & \alpha d_i\!\!-\!\!c_i\!
\end{array}
\!\!\right],
\mathbf{P}_2\!\!\triangleq\!\!
\left[\!\!
\begin{array}{cc}
\mathbf{C}_i & \mathbf{0} \\
\mathbf{0}^T & 0
\end{array}
\!\!\right],
\mathbf{P}_3\!\!\triangleq\!\!
\left[\!\!
\begin{array}{cc}
\mathbf{O} & \mathbf{0} \\
\mathbf{0}^T & 1
\end{array}
\!\!\right].\nonumber
\end{align}
\end{subequations}

Since $\mathbf{C}_i\succ0$, the feasible set of ($\mathsf{P}10_{\alpha}^i$) is bounded. Obviously the objective of ($\mathsf{P}10_{\alpha}^i$) takes finite value over the feasible set, so ($\mathsf{P}10_{\alpha}^i$) is solvable by Weierstrass’s theorem(proposition 3.2.1-(1) in \cite{bib:convex_optimization_theory}).

Assume that $\widetilde{\mathbf{X}}^{\star}$ is one optimal solution. Obviously $\widetilde{\mathbf{X}}^{\star}$ is non-zero(otherwise constraint $\mathbf{P}_3\widetilde{\mathbf{X}}^{\star}=1$ would fail). Since $\widetilde{\mathbf{X}}^{\star}$ has dimension $K_iN_i+1\geq3$, evoking theorem 2.2 of \cite{bib:rank_one_decomp}, we can obtain a vector $\widetilde{\mathbf{x}}$ such that $\Tr\{\mathbf{P}_j\widetilde{\mathbf{X}}^{\star}\}=\Tr\{\mathbf{P}_j\widetilde{\mathbf{x}}\widetilde{\mathbf{x}}^H\}$ for $j=1,2,3$.
Denote $\widetilde{\mathbf{x}}=[\widetilde{\mathbf{x}}_1^T, \tilde{x}_2]^T$. Notice that $\tilde{x}_2$ is nonzero(otherwise the constraint $\Tr\{\mathbf{P}_3\widetilde{\mathbf{x}}\widetilde{\mathbf{x}}^H\}=|\tilde{x}_2|^2=1$ would fail). Define $\widehat{\mathbf{x}}\triangleq[\widetilde{\mathbf{x}}_1^T/\tilde{x}_2, 1]^T$, it is easy to check that 
\begin{align}
\!\!\!\!f_0\big(\frac{\widetilde{\mathbf{x}}_1}{\tilde{x}_2}\big)&\!\!=\!\!\Tr\{\mathbf{P}_1\widehat{\mathbf{x}}\widehat{\mathbf{x}}^H\}\!\!=\!\!\Tr\{\mathbf{P}_1\widetilde{\mathbf{x}}\widetilde{\mathbf{x}}^H\}\!\!=\!\!\Tr\{\mathbf{P}_1\widetilde{\mathbf{X}}^{\star}\}\!\!=\!\!\mathsf{opt}(\mathsf{P}10_{\alpha}^i);\nonumber\\
\!\!\!\!f_1\big(\frac{\widetilde{\mathbf{x}}_1}{\tilde{x}_2}\big)&\!\!=\!\!\Tr\{\mathbf{P}_2\widehat{\mathbf{x}}\widehat{\mathbf{x}}^H\}\!\!=\!\!\Tr\{\mathbf{P}_2\widetilde{\mathbf{x}}\widetilde{\mathbf{x}}^H\}\!\!=\!\!\Tr\{\mathbf{P}_2\widetilde{\mathbf{X}}^{\star}\}\!\leq\! P_i.
\end{align}
where $f_0(\cdot)$ denotes the objective function of ($\mathsf{P}8_{\alpha}^i$) and $f_1(\mathbf{x})\triangleq\mathbf{x}^H\mathbf{C}_i\mathbf{x}$.  
The above two equations imply that $\widetilde{\mathbf{x}}_1/\tilde{x}_2$ is an optimal solution to ($\mathsf{P}8_{\alpha}^i$) since which gives optimal value $\mathsf{opt}(\mathsf{P}10_{\alpha}^i)$ and is feasible.
\end{proof}

\subsection{Proof of Theorem \ref{thm:closedform_P_i_D1}}
\label{subsec:appendix_thm_closedform_P_i_D1}

\begin{proof}
When $K_i=1$, the covariance matrices $\mathbf{\Sigma}_{\mathbf{s}}$ and $\{\mathbf{\Sigma}_i\}_{i=1}^L$ become scalars $\sigma_{s}^2$ and $\{\sigma_i^2\}_{i=1}^L$ respectively and we have
\begin{subequations}
\begin{align}
\mathbf{A}_{ii}&=\mathbf{H}_i^H\mathbf{g}\mathbf{g}^H\mathbf{H}_i,\ \ \ \mathbf{B}_i=\sigma_i^2\mathbf{H}_i^H\mathbf{g}\mathbf{g}^H\mathbf{H}_i,\\
\mathbf{C}_i&=(\sigma_i^2\!\!+\!\!\sigma_s^2)\mathbf{I}_{N_i}, \\
\mathbf{A}_{ij}&=\mathbf{1}_{K_j}^T\otimes\big(\mathbf{H}_i^H\mathbf{g}\mathbf{g}^H\mathbf{H}_i\big)=\mathbf{H}_i^H\mathbf{g}\big(\mathbf{1}_{K_j}^T\otimes\mathbf{g}^H\mathbf{H}_j\big), \\
\mathbf{q}_i&=\sum_{j\neq i}\mathbf{A}_{ij}\mathbf{f}_j=\mathbf{H}_i^H\mathbf{g}\Big[\sum_{j\neq i}\big(\mathbf{1}_{K_j}^T\otimes\mathbf{g}^H\mathbf{H}_j\big)\mathbf{f}_j\Big]. 
\end{align}
\end{subequations}
To simplify the following discussion, we introduce the notations
\begin{align}
\label{eq:def_P_bar}
\beta_i&\triangleq\sum_{j\neq i}\mathbf{f}_j^H\Big(\mathbf{1}_{K_j}\otimes\mathbf{H}_j^H\mathbf{g}\Big),\ \bar{P}_i\triangleq\frac{P_i}{\sigma_i^2+\sigma_s^2}.
\end{align}

Then the problem ($\mathsf{P}8_{\alpha}^i$) in (\ref{eq:opt_prob_P_i}) is expressed as
\begin{subequations}
\begin{align}
&\!\!\!\!\!\!\!\!\!(\mathsf{P}11_{\alpha}^i):\underset{\mathbf{f}_i}{\min.}\ \big(\alpha\sigma_i^2\!-\!1\big)\mathbf{f}_i^H\mathbf{H}_i^H\mathbf{g}\mathbf{g}^H\mathbf{H}_i\mathbf{f}_i\!-\!2\mathsf{Re}\big\{\beta_i\mathbf{g}^H\mathbf{H}_i\mathbf{f}_i\big\}\nonumber\\
&\qquad\qquad\quad+\!\big(\alpha d_i\!-\!c_i\big),\\
&\qquad \mathrm{s.t.}\ \|\mathbf{f}_i\|^2\leq \bar{P}_i.
\end{align}
\end{subequations}

The key observation is that the quadratic matrix  $\mathbf{H}_i^H\mathbf{g}\mathbf{g}^H\mathbf{H}_i$ in the objective function has rank one and thus an eigenvalue decomposition as follows
\begin{align}
\!\!\!\!\mathbf{H}_i^H\mathbf{g}\mathbf{g}^H\mathbf{H}_i\!=\!\mathbf{U}
\!\!\left[
\begin{array}{cc}
\mathbf{g}^H\mathbf{H}_i\mathbf{H}_i^H\mathbf{g} &  \\
 & \mathbf{O}_{(N_i-1)\times(N_i-1)}
\end{array}
\!\!\!\!\right]\mathbf{U}^H,\label{eq:rank_1_eigenDcmps_1}
\end{align}
with $\mathbf{U}\!\triangleq\!\big[\mathbf{u}_1, \mathbf{u}_2,\cdots,\mathbf{u}_{N_i}\big]$ being eigenvectors of $\mathbf{H}_i^H\mathbf{g}\mathbf{g}^H\mathbf{H}_i$. The first eigenvector $\mathbf{u}_1$ corresponds to the unique nonzero eigenvalue and the other eigenvectors span the null space of $\mathbf{H}_i^H\mathbf{g}\mathbf{g}^H\mathbf{H}_i$. 
In other words, we have
\begin{align}
\mathbf{u}_1=\frac{\mathbf{H}_i^H\mathbf{g}}{\|\mathbf{H}_i^H\mathbf{g}\|_2}, \ \ \ \ \mathbf{u}_j^H\mathbf{H}_i^H\mathbf{g}=0, \ \ \ j=\{2,\cdots,N_i\}\label{eq:rank_1_eigenDcmps_2}
\end{align}
Since $\{\mathbf{u}_i\}_{i=1}^L$ is an orthonormal basis, $\mathbf{f}_i$ can be represented as $\mathbf{f}_i=\mathbf{U}\mathbf{\tau}=\sum_{j=1}^{N_i}\mathbf{u}_j\tau_j$ with vector $\mathbf{\tau}$ being the coordinates in terms of basis $\{\mathbf{u}_j\}_{j=1}^{N_i}$.

By (\ref{eq:rank_1_eigenDcmps_2}) the objective of ($\mathsf{P}11_{\alpha}^i$) is independent of $\{\tau_j\}_{j=2}^{N_i}$. To save power, we should set all $\{\tau_j\}_{j=2}^{N_i}$ as zero, which means $\mathbf{f}_i=\tau_1\frac{\mathbf{H}_i^H\mathbf{g}}{\|\mathbf{H}_i^H\mathbf{g}\|_2}$. Thus, the problem ($\mathsf{P}11_{\alpha}^i$) boils down to the following problem with respect to one complex scalar $\tau_1$
\begin{subequations}
\begin{align}
&\!\!\!\!\!\!(\mathsf{P}12_{\alpha}^i):\underset{\tau_1}{\min.}\ g(\tau_1)\!\triangleq\!(\alpha\sigma_i^2\!-\!1)\|\mathbf{H}_i^H\mathbf{g}\|_2^2|\tau_1|^2 \\
&\qquad\qquad\qquad\qquad-\!2\|\mathbf{H}_i^H\mathbf{g}\|_2\mathsf{Re}\{\beta_i\tau_1\}\!+\!\big(\alpha d_i\!-\!c_i\big),\nonumber\\
&\qquad\quad\ \mathsf{s.t.} \ |\tau_1|^2 \leq \bar{P}_i. \label{eq:opt_const_P6i}
\end{align}
\end{subequations}
Based on the sign of $(\sigma_i^2-\alpha)$, the problem ($\mathsf{P}12_{\alpha}^i$) can be tackled in the following three cases:\newline
\underline{$\mathsf{CASE}$ (\uppercase\expandafter{\romannumeral1}): $\alpha=\sigma_i^{-2}$.}
In this case, the objective function in $(\mathsf{P}12_{\alpha}^i)$ degenerates to an affine function
\begin{align}
g(\tau_1)=-2\sigma_i^2\|\mathbf{H}_i^H\mathbf{g}\|_2\mathsf{Re}\{\beta_i\tau_1\}+(\alpha d_i-c_i).
\end{align}
By the Cauchy-Schwarz inequality, the optimal $\tau_1^{\star}$ and minimum objective is  obtained as
\begin{align}
\!\!\!\!\!\!\tau_1^{\star}&\!=\!\frac{\sqrt{\bar{P}_i}\beta_i^{\ast}}{|\beta_i|}; g(\tau_1^{\star})\!=\!-2\sigma_i^2\|\mathbf{H}_i\mathbf{g}\|_2|\beta_i|\sqrt{\bar{P}_i}\!\!+\!\!\big(\alpha d_i\!-\!c_i\big);\label{eq:opt_sol_P12i_Case1}
\end{align}

When $\alpha\neq\sigma_i^{-2}$ by denoting
\begin{align}
\zeta_i\triangleq&\big(\alpha d_i\!-\!c_i\big)\!-\!\frac{|\beta_i|^2}{\alpha\sigma_i^2-1},
\end{align}
the objective function is equivalently written as
\begin{align}
\!\!\!\!g(\tau_1)\!=\!(\alpha\sigma_i^2\!\!-\!\!1)\|\mathbf{H}_i^H\mathbf{g}\|_2^2\bigg|\tau_1\!-\!\frac{\beta_{i}^{\ast}}{\|\mathbf{H}_i^H\mathbf{g}\|_2(\alpha\sigma_i^2\!-\!1)}\bigg|^2\!\!+\!\!\zeta_i,\label{eq:opt_prob_P12i_2}
\end{align}
\noindent \underline{$\mathsf{CASE}$ (\uppercase\expandafter{\romannumeral2}):  $\alpha>\sigma_i^{-2}$}.
To minimize $g(\tau_1)$, $\tau_1$ should be along the direction of $\beta_i^*$. Depending on whether the zero point of the absolute term in (\ref{eq:opt_prob_P12i_2}) satisfies the power constraint, two subcases are examined:\newline
\indent\expandafter{\romannumeral1})
If $\frac{|\beta_i^{\ast}|}{\|\mathbf{H}_i^H\mathbf{g}\|_2(\alpha\sigma_i^2\!-\!1)}\leq\sqrt{\bar{P}_i}$, the optimum is given as
\label{eq:opt_sol_P12i_Case2_1}
\begin{align}
\tau_1^{\star}\!&=\!\frac{\beta_i^{\ast}}{\|\mathbf{H}_i^H\mathbf{g}\|_2(\alpha\sigma_i^2\!\!-\!\!1)}; \ \ \ g(\tau_1^{\star})\!=\!\zeta_i; 
\end{align}
\indent\expandafter{\romannumeral2}) If $\frac{|\beta_i^{\ast}|}{\|\mathbf{H}_i^H\mathbf{g}\|_2(\alpha\sigma_i^2\!-\!1)}>\sqrt{\bar{P}_i}$, the optimum is given as
\begin{align}
\label{eq:opt_sol_P12i_Case2_2}
\tau_1^{\star}&=\frac{\sqrt{\bar{P}_i}\beta_i^{\ast}}{|\beta_i|};\\
g(\tau_1^{\star})&\!\!=\!\!(\alpha\sigma_i^2\!\!-\!\!1)\|\mathbf{H}_i^H\mathbf{g}\|_2^2\bigg|\sqrt{\bar{P}_i}\!-\!\frac{|\beta_{i}|}{\|\mathbf{H}_i^H\mathbf{g}\|_2(\alpha\sigma_i^2\!-\!1)}\bigg|^2\!\!+\!\!\zeta_i\nonumber
\end{align}

\noindent\underline{$\mathsf{CASE}$ (\uppercase\expandafter{\romannumeral3}): $\alpha<\sigma_i^{-2}$.}
Still $\tau_1$ should be along the direction of $\beta_i^{\ast}$, and takes full power. At this time the optimum is literally identical with (\ref{eq:opt_sol_P12i_Case2_2}) in the above.

Remember that optimal $\mathbf{f}_i^{\star}$ to ($\mathsf{P}8_{\alpha}^i$) is obtained by $\mathbf{f}_i^{\star}=\tau_1^{\star}\frac{\mathbf{H}_i^H\mathbf{g}}{\|\mathbf{H}_i^H\mathbf{g}\|_2}$, the proof is complete.
\end{proof}

\end{document}